\documentclass[a4paper,onecolumn,10pt,unpublished]{quantumarticle}
\pdfoutput=1
\usepackage[utf8]{inputenc}
\usepackage[english]{babel}
\usepackage[T1]{fontenc}
\usepackage{amsmath}
\usepackage{amssymb}
\usepackage{amsthm}
\usepackage[colorlinks=false,hidelinks]{hyperref}
\usepackage{braket}
\usepackage{cleveref}
\usepackage{centernot}
\usepackage{enumerate}
\usepackage{thmtools}
\usepackage{thm-restate}
\usepackage{subcaption} 
\usepackage{suffix}

\usepackage{tikz}
\usepackage{lipsum}

\usepackage[numbers,sort&compress]{natbib}

\usepackage{printlen}

\newcommand{\cH}{\mathcal{H}}
\newcommand{\U}[1]{\mathcal{U}(#1)}
\newcommand{\PU}[1]{\mathcal{PU}(#1)}
\DeclareMathOperator{\Tr}{Tr}

\DeclareMathOperator{\argmin}{argmin}
\newcommand{\imp}{\Lambda}

\newcommand{\lanom}{\Omega}
\newcommand{\maxent}{\Phi}
\newcommand{\fsU}[2]{\gamma_{U,#2}(#1)}
\newcommand{\Wset}[1]{\mathcal{W}_{U,#1}}

\newcommand{\arb}{\alpha}

\declaretheorem[name=Lemma]{lemma}

\theoremstyle{definition}

\newcommand{\circone}{\emph{no-entanglement} }
\WithSuffix\newcommand\circone*{\emph{no-entanglement}}
\newcommand{\circfour}{\emph{$CR_X$-entanglement} }
\WithSuffix\newcommand\circfour*{\emph{$CR_X$-entanglement}}
\newcommand{\circnine}{\emph{$CZ$-entanglement} }
\WithSuffix\newcommand\circnine*{\emph{$CZ$-entanglement}}
\newcommand{\circfifteen}{\emph{circular-entanglement} }
\WithSuffix\newcommand\circfifteen*{\emph{circular-entanglement}}

\begin{document}

\title{Loss Behavior in Supervised Learning with Entangled States}

\author{Alexander Mandl}
\orcid{0000-0003-4502-6119}
\email{mandl@iaas.uni-stuttgart.de}
\author{Johanna Barzen}
\orcid{0000-0001-8397-7973}
\email{barzen@iaas.uni-stuttgart.de}
\author{Marvin Bechtold}
\orcid{0000-0002-7770-7296}
\email{bechtold@iaas.uni-stuttgart.de}
\author{Frank Leymann}
\orcid{0000-0002-9123-259X}
\email{leymann@iaas.uni-stuttgart.de}
\author{Lavinia Stiliadou}
\orcid{0009-0001-0957-6108}
\email{stiliadou@iaas.uni-stuttgart.de}
\affiliation{Institute of Architecture of Application Systems, University of Stuttgart, Universitätsstraße 38, 70569 Stuttgart, Germany}

\maketitle

\begin{abstract}
Quantum Machine Learning (QML) aims to leverage the principles of quantum mechanics to speed up the process of solving machine learning problems or improve the quality of solutions.
Among these principles, entanglement with an auxiliary system was shown to increase the quality of QML models in applications such as supervised learning.
Recent works focus on the information that can be extracted from entangled training samples and their effect on the approximation error of the trained model.
However, results on the trainability of QML models show that the training process itself is affected by various properties of the supervised learning task.
These properties include the circuit structure of the QML model, the used cost function, and noise on the quantum computer.
To evaluate the applicability of entanglement in supervised learning, we augment these results by investigating the effect of highly entangled training data on the model's trainability.
In this work, we show that for highly expressive models, i.e., models capable of expressing a large number of candidate solutions, the possible improvement of loss function values in constrained neighborhoods during optimization is severely limited when maximally entangled states are employed for training. 
Furthermore, we support this finding experimentally by simulating training with Parameterized Quantum Circuits (PQCs).
Our findings show that as the expressivity of the PQC increases, it becomes more susceptible to loss concentration induced by entangled training data. 
Lastly, our experiments evaluate the efficacy of non-maximal entanglement in the training samples and highlight the fundamental role of entanglement entropy as a predictor for the~trainability.
\end{abstract}

\section{Introduction}\label{sec:intro}
Quantum Machine Learning~(QML) uses quantum computers or quantum-inspired methods to solve machine learning problems~\cite{Schuld2018,Biamonte2017,Huang2021}.
In QML, quantum computers are applied to various tasks, such as classification~\cite{Rebentrost2014,Farhi2018} or combinatorial optimization~\cite{Farhi2014,Egger_2021}.
Furthermore, the applicability of proven methods from classical machine learning, such as supervised learning of transformations using labeled training data on quantum computers, is explored~\cite{Schuld2018,Poland2020,Sharma2020,Khatri2019}.
In this work, we study supervised learning of quantum operators.
Herein, the aim is to approximate an unknown quantum operator, the \emph{target operator}, using a set of inputs with their associated outputs of the target operator as training samples.
To achieve this goal, a QML model such as a Parameterized Quantum Circuit~(PQC)~\cite{Benedetti2019a} is continuously adapted until its outputs, when applied to the training sample inputs, match the target operator's outputs as accurately as possible. 
This approach can be applied, for example, to reduce the circuit complexity of the approximated operator~\cite{Sharma2020} or to learn the dynamics of a physical process on a quantum computer~\cite{Huang2022,Caro2023,Volkoff2021}.

During the training process in supervised learning, a classical optimizer minimizes a loss function that quantifies the deviation between the model's output and the expected output, as determined by the given training samples. 
After concluding the training process, the quality of a learned transformation can be measured as the average loss over all possible, even unseen, input states and their outputs according to the target transformation~\cite{Poland2020,Sharma2022,Caro2023}.
This quantity is referred to as the \emph{risk}. 
The composition of the training samples plays a major role in determining the risk of a learned transformation.
For example, it was shown that, under certain conditions, introducing entanglement with an auxiliary system, the \emph{reference system}, can improve the risk~\cite{Sharma2022,Wang2023}.
Particularly, repeated access to a single maximally entangled training sample provides enough information about a unitary operator to reproduce it exactly~\cite{Sharma2022}. 
Furthermore, even non-maximal entanglement is beneficial as increasing the degree of entanglement of a training sample, as measured by its Schmidt rank, decreases the best-case risk after training if certain mathematical properties of the training samples are adhered to~\cite{Mandl2023reducing}. 

While these results suggest that introducing entanglement is desirable for improving the risk, they primarily focus on the achievable performance after training, assuming the training process successfully finds a suitable solution.
However, the training process itself is affected by various factors in the problem setup.~%
These include the expressivity of the PQC that is trained~\cite{Holmes2022}, the type of readout operators employed~\cite{Larocca2025,Cerezo2021a}, or possible noise in the circuit execution~\cite{Wang2021}.
Furthermore, the training process can be affected by the entanglement among the qubits within the training samples themselves, even without introducing an external reference system~\cite{Leone2024,Thanasilp2023}.
Thus, although the risk is reduced when entanglement with a reference system is employed, this improvement might come with an increase in training complexity. 
Therefore, in this work, we aim to evaluate the complexity of the training process when entanglement with a reference system is employed. 
Recent works have evaluated the training complexity by investigating the structure of the landscape described by the loss function~\cite{Stiliadou2024,Larocca2025,Holmes2022,Arrasmith2022}.
Certain features, such as flat regions in the loss landscape, are suspected to be detrimental to the optimization performance, as they hinder the optimizer from discerning the optimal direction for optimization~\cite{Larocca2025}.
In particular, loss landscapes can exhibit \emph{barren plateaus}, wherein the loss differences diminish exponentially with the number of qubits~\cite{Mhiri2025,Larocca2025}.

While barren plateaus are generally described as a global phenomenon of the loss function, recent works evaluate local neighborhoods in the loss landscape~\cite{Mhiri2025,Mastropietro2023}.
This allows for specifying regions of the loss landscape that are suitable for optimization, even though the loss function might exhibit flat regions in general. 
We employ this approach in this work by evaluating the concentration of the loss function values in local neighborhoods given by a metric on the set of admissible operators. 
In particular, we quantify the largest possible improvement in loss function values based on the training sample that is used.~%
By abstracting the optimization procedure to the case of optimizing over all unitary operators, we show that the loss improvement decreases exponentially with the number of qubits when a maximally entangled sample is used.
Furthermore, we evaluate this result experimentally using a selection of PQCs. 
These PQCs are trained while restricting the admissible solutions to a neighborhood of the starting point of the optimization. 
These experiments show that highly expressive PQCs are negatively affected by the high degrees of entanglement in the training samples the most.
Lastly, our experiments suggest that this effect can be mitigated when non-maximally entangled samples of high Schmidt rank are used.
This presents a promising avenue for selecting suitable training samples for quantum supervised learning.

In the remainder of this paper, we first introduce the required concepts for our results in \Cref{sec:background} and present an example as motivation, highlighting the effect of entanglement on the loss landscape in \Cref{sec:motivation}. 
We proceed by presenting our analytical results on the effect of maximally entangled training samples on the loss function in \Cref{sec:analytical}.
These results, as well as the effect of non-maximal entanglement, are experimentally evaluated in \Cref{sec:exp} and discussed in \Cref{sec:discussion}.
Lastly, we present related work on the role of entanglement in supervised learning in \Cref{sec:rel_work} and summarize our results in \Cref{sec:conclusion}.

\section{Background}\label{sec:background}
This section provides an overview of the mathematical preliminaries relevant to this work and presents the supervised learning problem.
To quantify the entanglement present in the training samples, we employ two different entanglement measures, the Schmidt rank and the entanglement entropy.
These measures are introduced in \Cref{sec:schmidt_decomp_majorization}.
In particular, the Schmidt rank was used in recent work to study the risk in supervised learning on a quantum computer~\cite{Sharma2020}.
The results are summarized in \Cref{sec:supervised_learning} along with the general problem description of supervised learning with entangled states. 
\Cref{sec:metrics} proceeds to define metrics on quantum states and on quantum circuits, which are used repeatedly for our analytical results.
Lastly, the training process in supervised learning requires a model that can be adjusted via classical parameters, and we use PQCs for this task. 
Thus, \Cref{sec:pqcs} presents a general description of PQCs and highlights the concept of expressivity.

\subsection{Entanglement in Pure Quantum States}\label{sec:schmidt_decomp_majorization}

Any pure quantum state $\ket{\arb}$ of a tensor product Hilbert space $\cH_X \otimes \cH_R$ with dimension $d = \max(\dim(\cH_X), \dim(\cH_R))$ can be expressed using its Schmidt decomposition~\cite{nielsen2002quantum}
\begin{align}
    \ket{\arb} = \sum_{j=1}^d \sqrt{c_j} \ket{x_j}_X \otimes \ket{y_j}_R\label{eq:schmidt_decomp},
\end{align}
with orthonormal \emph{Schmidt bases} $\{\ket{x_j}\}_{j=1}^d \subseteq \cH_X$, $\{\ket{y_j}\}_{j=1}^d \subseteq \cH_R$ and $c_j \in \mathbb{R}_{\geq 0}$ with the normalization condition $\sum_{j=1}^d c_j = 1$.
The individual $\sqrt{c_j}$ are referred to as the \emph{Schmidt coefficients}, and the number $r$ of nonzero Schmidt coefficients defines the \emph{Schmidt rank}~\cite{nielsen2002quantum}.
The reference system $\cH_R$ is required to be large enough to express $r$ basis states. 
If not specified otherwise, we use a reference system of the same dimension as $\cH_X$ and refer to their dimension as $d = \dim(\cH_X) = \dim(\cH_R)$.
Furthermore, for an $n$-qubit system, $d=2^n$.
A state $\ket{\maxent} \in \cH_X \otimes \cH_R$ of maximal Schmidt rank with equal Schmidt coefficients is \emph{maximally entangled}:
\begin{align}
    \ket{\maxent} = \frac{1}{\sqrt{d}} \sum_{j=1}^d\ket{x_j}_X \otimes \ket{y_j}_R. \label{eq:def_max_ent} 
\end{align}

While the Schmidt rank is solely based on the number of nonzero Schmidt coefficients, other entanglement measures take the values of the Schmidt coefficients into account.
One such measure used in our experiments is the \emph{entanglement entropy} $E(\ket{\arb})$~\cite{Bennett1996}.
For a pure state $\ket{\arb} \in \cH_X \otimes \cH_R$, it is defined as the von Neumann entropy $E(\ket{\arb}) := S(\rho_X) = -\Tr(\rho_X\ln(\rho_X))$ of the operator $\rho_X = \Tr_R(\ket{\arb}\!\bra{\arb})$ after tracing out the reference system~\cite{Bennett1996,Bengtsson2017}.
Using \Cref{eq:schmidt_decomp}, the entanglement entropy is calculated as the entropy of the squared Schmidt coefficients $c_j$~\cite{Bennett1996}
\begin{align}
    E(\ket{\arb}) = - \sum_{j=1}^d c_j \ln(c_j)\label{eq:ent_entropy}, 
\end{align}
and is minimal at $E(\ket{\arb}) = 0$ for separable states and maximal at $E(\ket{\arb}) = \ln(d)$ for maximally entangled states.
Furthermore, entangled quantum states with $0 < E(\ket{\arb}) < E(\ket{\maxent})$ are referred to as \emph{non-maximally entangled}~(NME).

Throughout this work, we denote separable quantum states as $\ket{\psi} \in \cH_X \otimes \cH_R$ and, if the factorization is irrelevant, we assume $\ket{\psi} \in \cH_X$.
Furthermore, in our results we denote any maximally entangled state (\Cref{eq:def_max_ent}), regardless of the used Schmidt basis, as $\ket{\maxent} \in \cH_X \otimes \cH_R$ and use $\ket{\arb}\in\cH_X \otimes \cH_R$ for results and definitions that apply to arbitrary states on the product space.
If the Hilbert space is clear from the context, we further omit the subscripts $X$ and $R$ denoting the respective systems in the factorization.

\subsection{Supervised Learning of Unitary Operators}\label{sec:supervised_learning}

The QML task evaluated in this work is the replication of unitary operators using supervised learning.
In this setup, a quantum circuit $V$ (the \emph{hypothesis operator}) is iteratively adapted during the training process to match an unknown target operator $U : \cH_X \to \cH_X$ on the Hilbert space $\cH_X$ by making use of training samples.
These training samples are tuples of quantum input states $\ket{\arb} \in \cH_X$ with their respective output $U\!\ket{\arb}$ after the application of the operator $U$.

The supervised learning process aims to find a hypothesis operator that matches the target operator on the training samples. 
In other words, the output states of the hypothesis operator should match those of the target operator when evaluated on the same input states.
As the loss function, a measure of dissimilarity between training sample output $U\!\ket{\arb}$ and hypothesis output $V\!\ket{\arb}$ given by the complement of the \emph{fidelity} $F(\ket{a}, \ket{b}) = \left| \braket{a|b} \right|^2$~\cite{nielsen2002quantum} is used:
\begin{align}
    L_{U,\arb}(V) = 1 - F(U\!\ket{\arb}, V\!\ket{\arb}).\label{eq:no_entanglement_loss}
\end{align}
In cases where multiple training samples are available, the average of this function is used.

Minimizing the loss function on all training samples, therefore, maximizes the fidelity of the respective outputs with respect to all training samples. 
Beyond minimizing the training loss, the goal is to ensure generalization of the hypothesis operator to unseen input states.
The generalization error is captured by the risk function~\cite{Sharma2022,Mandl2023reducing} defined as
\begin{align}\label{eq:risk_function}
    R_U(V) = \mathbb{E}_{\ket{\arb}}\left[ L_{U,\arb} (V) \right].
\end{align}
Herein, the expectation value is taken with respect to the Haar measure~\cite{Mele2023,Puchala_2017} on the set of quantum states $\ket{\arb} \in \cH_X$.
Thus, the risk function measures the expected loss when averaged over all possible inputs to the target operator.

For this training setup, it was shown that introducing entanglement in the training samples is beneficial for minimizing the risk of the resulting operator~\cite{Sharma2022,Volkoff2021}.
This is achieved through entanglement with an auxiliary system $\cH_R$.
Thus, the entangled training sample inputs are elements of the product space $\cH_X \otimes \cH_R$.
The expected outputs for each training sample are given by applying~$U$ to $\cH_X$ and leaving the auxiliary system unchanged: $(U_X \otimes I_R)\ket{\arb}$.
Similarly, during training, the hypothesis operator is applied to $\cH_X$ only, which leads to the adapted loss function
\begin{align}\label{eq:loss_function}
    L_{U,\arb}(V) = 1 - F((U_X \otimes I_R)\!\ket{\arb}, (V_X \otimes I_R)\!\ket{\arb}).
\end{align}
In particular, for $d=\dim(\cH_X)$, sets of training samples $S$ of size $t$ and the degree of entanglement given by the Schmidt rank $r$ (see also \Cref{sec:schmidt_decomp_majorization}) of each sample, Sharma~et~al.~\cite{Sharma2022} showed
\begin{align}
    \mathbb{E}_U \left[ \mathbb{E}_S \left[ R_U(V) \right] \right] \geq 1- \frac{r^2 t^2 +d +1}{d(d+1)}.\label{eq:qnfl_theorem}
\end{align}
The expectation value is taken over the Haar measure of all unitary operators acting on $\cH_X$ and all possible sets of training samples of size $t$~\cite{Sharma2022}.
This result implies that increasing the Schmidt rank of the training samples reduces the lower bound for the risk in the same way as increasing the number of training samples does. 
In particular, if the training loss of the hypothesis $V$ is zero and a training sample of maximal Schmidt rank ($r=d$) is used, then $R_U(V) = 0$, i.e., it is possible that the hypothesis perfectly replicates the target operator~\cite{Sharma2022, Mandl2023reducing}.

\subsection{Metrics on States and Operators}\label{sec:metrics}

The loss function $L_{U,\arb}(V)$ is defined as the infidelity of the expected output of the target transformation $(U_X \otimes I_R)\ket{\arb}$ and the actual output $(V_X \otimes I_R) \ket{\arb}$ of the model $V$ that is trained.
The fidelity is related to a distance measure on the set of quantum states, 
the \emph{Fubini-Study distance}~\cite{Bengtsson2017}
\begin{align}
    \gamma(\ket{a}, \ket{b}) := \arccos\left( \sqrt{\frac{\left|\braket{a|b}\right|^2}{\braket{a|a}\braket{b|b}}}\right) = \arccos\left( \sqrt{\frac{F(\ket{a},\ket{b})}{\braket{a|a}\braket{b|b}}}\right)\label{eq:fs_distance}
\end{align}
on the space $\mathbb{C}\mathrm{P}^{d-1}$, the space of vectors in $\mathbb{C}^d$ under equivalence regarding global phases.
The explicit normalization of the inner product in the definition in \Cref{eq:fs_distance} is omitted if the arguments are assumed to be normalized quantum states.
This distance measure can be interpreted as the angle between two quantum states~\cite{nielsen2002quantum} and is also referred to as the \emph{Bures angle} or \emph{quantum angle}~\cite{Wu2020}.
It takes the value $\gamma(\ket{a},\ket{b}) = 0$ for states that are equal up to global phase differences and $\gamma(\ket{a},\ket{b}) = \pi/2$ for orthogonal states.

Let $F_{U,\arb}(V) := F((U_X \otimes I_R)\!\ket{\arb}, (V_X \otimes I_R)\!\ket{\arb}) = 1-L_{U,\arb}(V)$ be a shorthand for the fidelity between expected and actual output and let $\fsU{V}{\arb} := \arccos\left(\sqrt{F_{U,\arb}(V)}\right)$ refer to the associated Bures angle.
The loss function in supervised learning is equivalently expressed using $\fsU{V}{\arb}$ as 
\begin{align}
    L_{U,\arb}(V) &= 1-F_{U,\arb}(V)\\
    &= 1-\cos(\fsU{V}{\arb})^2\\
    &= \sin(\fsU{V}{\arb})^2.\label{eq:loss_sinus}
\end{align}
Therefore, during training, the optimizer aims to decrease the Bures angle between the model output and the expected output.
Because of this close relationship between the fidelity $F_{U,\arb}(V)$, the loss function $L_{U,\arb}(V)$, and the Bures angle $\fsU{V}{\arb}$, these reformulations are used repeatedly in~\Cref{sec:analytical}.

A supervised learning algorithm aims to find a unitary operator that minimizes this loss function with respect to the given training samples. 
Since the loss function is invariant under the multiplication of $V$ with a global phase factor of the form $e^{i\rho}$, this learning problem is therefore a task of finding a suitable unitary operator up to global phase differences. 
Thus, although the solution set is the set of all unitary operators $\U{d}$, we compare them in our analytical results by a metric on the set $\PU{d}$, which is the set of unitary operators with equivalence under multiplication of global phase factors~\cite{Hall2013}.
On this set, metrics can be defined by minimization over these global phase factors
\begin{align}
    d'(U,V) = \min_{\rho \in \mathbb{R}} d(U,e^{i\rho}V),
\end{align}
using a suitable metric $d(U,V)$, such as the Frobenius norm distance, on the set of unitary operators~\cite{Haah2023}.
This metric, in the following referred to as $d_F'(U,V)$, is derived from the Frobenius norm $\lVert X\rVert_F = \sqrt{\braket{X,X}_F} = \sqrt{\Tr(X^\dagger X)}$.
Thus,
\begin{align}
    d_F'(U,V) &:= \min_{\rho \in \mathbb{R}} d_F(U,e^{i\rho}V)\\
    &= \min_{\rho \in \mathbb{R}} \lVert U-e^{i\rho}V\rVert_F  \\
    &= \min_{\rho \in \mathbb{R}} \sqrt{ \Tr(U^\dagger U) + \Tr(V^\dagger V) - e^{i\rho} \Tr(U^\dagger V) - e^{-i\rho} \Tr(V^\dagger U)}\\
    &= \min_{\rho \in \mathbb{R}} \sqrt{ 2d - e^{i\rho} \Tr(U^\dagger V) - \overline{e^{i\rho} \Tr(U^\dagger V)} } \\
    &= \min_{\rho \in \mathbb{R}} \sqrt{2d - 2\Re(e^{i\rho}\Tr(U^\dagger V))}\\
    &= \sqrt{2d - 2\max_{\rho \in \mathbb{R}}\Re(e^{i\rho}\Tr(U^\dagger V))}\\
    &= \sqrt{2d}\sqrt{1-\frac{1}{d} \left| \Tr(U^\dagger V)\right|}\label{eq:rel_dist_trace},
\end{align}
where the last line follows since $\Re(e^{i\rho}z)$ is maximal at $\Re(e^{i\rho}z) = |z|$ for every complex number $z$.
This distance is minimal for unitary operators that only differ by a global phase with $d_F'(U,V) = 0$, and it is maximal for orthogonal operators with $d_F'(U,V) = \sqrt{2d}$.
By applying the definition of the maximally entangled state $\ket{\maxent}$ (\Cref{eq:def_max_ent}), a similar expression is obtained for the fidelity in the supervised learning problem
\begin{align}
    F_{U,\maxent}(V) &= \left| \braket{\maxent| U^\dagger V \otimes I|\maxent}\right|^2\\
    &= \left| \frac{1}{d} \sum_{j,k = 1}^d \braket{x_j| U^\dagger V | x_k} \braket{y_j | y_k} \right|^2\\
    &= \frac{1}{d^2}\left| \sum_{j=1}^d \braket{x_j | U^\dagger V | x_j}\right|^2\\
    &= \frac{1}{d^2} \left| \Tr(U^\dagger V)\right|^2\label{eq:fidelity_max_ent_trace}.
\end{align}
Therefore, 
\begin{align}
    d_F'(U,V) &= \sqrt{2d}\sqrt{1-\sqrt{F_{U,\maxent}(V)}}\label{eq:frob_relationship_fidelity}\\
    &= \sqrt{2d}\sqrt{1-\sqrt{1- L_{U,\maxent}(V)}}\label{eq:frob_relationship_loss}\\
    &= \sqrt{2d}\sqrt{1-\cos(\fsU{V}{\maxent})}\label{eq:frob_relationship}.
\end{align}
Thus, a reduction in the Bures angle $\fsU{V}{\maxent}$ between expected output and actual output when training with a maximally entangled state implies a reduction in the Frobenius norm distance between $U$ and $V$ after correcting for global phase factors. 
Furthermore, any operator $V$ with $\fsU{V}{\maxent} = 0$ has $d_F'(U,V) = 0$.
Thus, when training with a maximally entangled state, a global minimum occurs exactly for the operators $V$ that match the target operator $U$ up to a global~phase.

\subsection{Parameterized Quantum Circuits}\label{sec:pqcs}
The supervised learning task requires a way of exploring the set of possible hypothesis operators~$V$ during training. 
This is usually achieved by using Parameterized Quantum Circuits (PQCs)~\cite{Schuld2018,Sim2019,Cerezo2021} to realize an adaptable hypothesis operator that depends on classical parameters $\vec{\theta}$.
Thus, the PQC is a function $V(\vec{\theta}) : \mathbb{P} \to \U{d}$, which maps elements $\vec{\theta}$ of its \emph{parameter space} $\mathbb{P}$ to quantum circuits.
Generally, this is done using a circuit, often called an \emph{ansatz}, containing parameterized quantum gates and unparameterized gates in a predefined structure.
The parameters of the PQC are real values and for a PQC with $p$ parameterized gates, the parameter space is a set $\mathbb{P} \subseteq \mathbb{R}^p$~\cite{Barzen2025}.
Depending on the type of parameterized gates applied in the PQC, e.g.~periodic rotation gates, the parameter space is a proper subset of $\mathbb{R}^p$.
However, since the exact extent of this subset depends on the PQC, we assume $\mathbb{P} = \mathbb{R}^p$ in the following.

The structure of ansatz circuits is a field of active study~\cite{Guo2024,Cerezo2021,Sim2019} as it influences which operators can be expressed by varying $\vec{\theta}$.
A PQC can, for example, be modeled after the information that is known about the structure of the QML problem~\cite{Farhi2014,Wang2024}, or by including knowledge about the quantum hardware that the circuit is executed on~\cite{Kandala2017,Cerezo2021}.
If no information about the problem is known in supervised learning, it is generally preferred that a PQC models a large number of unitary operators. 
This increases the chances that the ansatz can express the solution to the supervised learning problem.
This concept of the expressive power of an ansatz is referred to as the \emph{expressivity}.
The evaluation and comparison of the expressivity of PQCs was studied via various methods in related work~\cite{Barzen2025}.
One approach is based on the evaluation of the dimension of the manifold given by the states that are reachable by the ansatz~\cite{Funcke2021,Barzen2025}.
Therein, a high dimension indicates a high expressivity of the PQC.
Another approach evaluates the PQC's ability to uniformly explore the Bloch sphere by adaptation of its parameters~\cite{Sim2019,Friedrich2023}.
We make use of the latter approach in this work as it can be experimentally approximated for a given PQC by sampling the distribution of output fidelities of the PQC (see also \Cref{app:expressivity}).
Following the description in~\cite{Sim2019}, this measure of expressivity is formally defined via the 
the Kullback-Leibler~(KL) divergence of the distribution of state fidelities $\hat{P}_V (F, \vec{\theta})$ of pairs of quantum states generated by the ansatz $V(\vec{\theta})$ and the distribution $P_{\text{Haar}}(F)$ of fidelities for randomly sampled quantum states
\begin{align}
    \mathrm{Expr} := D_{KL}\left( \hat{P}_V (F, \vec{\theta})\;\lVert\; P_{\text{Haar}}(F)\right).
\end{align}
Herein, low values for $\mathrm{Expr}$ imply a small divergence of the distributions and suggest a high expressive power of the ansatz.

To allow influencing the expressivity of a PQC, it is often structured in layers~\cite{Farhi2014,Sim2019}
\begin{align}
    V(\vec{\theta}) = \prod_{i=1}^l V_i(\vec{\theta}).
\end{align}
Herein, the layers $V_i$ usually contain the same quantum gates, and each layer uses a subset of the supplied parameters $\vec{\theta}$.
Thus, each layer differs only by the precise parameters it uses.
By increasing the number of layers $l$, the number of parameters $p$ is increased.
This is done with the aim of increasing the set of unitary operators that the PQC can express at the cost of circuit depth.
However, the actual increase in expressivity depends on the structure of the individual layers and has been shown to saturate after a finite number of layers for some PQCs~\cite{Sim2019}.

The loss function in supervised learning \Cref{eq:loss_function} defines a graph on the parameter space, the \emph{loss landscape}~\cite{Stiliadou2024,Wu2023}.
The geometric properties of this landscape influence the performance of the optimization process of the PQC.
For example, flat regions in the loss landscape, so-called \emph{barren plateaus}, hinder the optimizer's ability to find an optimum during training~\cite{Larocca2025}.
In particular, higher expressivity in PQCs has been linked to an increased likelihood of barren plateaus in the loss landscape~\cite{Holmes2022}.
Barren plateaus are commonly characterized using global properties of the loss landscape, such as the variance of loss gradients~\cite{Pesah2021,McClean2018,Holmes2022,Grant2019} or the concentration of the loss function values over the parameter space~\cite{Arrasmith2022,Wang2021,Crognaletti2024}.
However, some recent works aim to treat barren plateaus as a local feature of the loss landscape. 
For example, to characterize the performance of the optimization approach in~\cite{Mastropietro2023}, a barren plateau is defined as a ball in the parameter space of the PQC where gradient magnitudes vanish. 
Similarly, in~\cite{Mhiri2025} small subregions of the loss landscape are considered to identify regions of sufficiently large gradients that can, for example, be used for inferring suitable initial parameters for the optimization.
Motivated by these approaches, we aim to evaluate the optimization complexity of supervised learning with entangled training samples using the variation of the loss function values in balls in the loss landscape in our analytical results in \Cref{sec:analytical} as well as in our experiments in \Cref{sec:exp}.

\section{Motivation and Problem Statement}\label{sec:motivation}

Results on the generalization error in supervised learning, such as the bounds presented in \Cref{eq:qnfl_theorem}, encompass the fact that the loss function used during training might not match the risk function. 
Even after a hypothesis operator of minimal loss is found, its risk might not be minimal.
Thus, there still could be quantum states that are mapped to incorrect outputs by the hypothesis operator.
This implies that there are global minima in the loss function that do not correspond to a global minimum in the risk function, as long as a sufficient number of training samples is not available.
Conversely, if a sufficient amount of training samples is available (e.g., in the form of a highly entangled training sample), the loss function increasingly matches the risk function. 
This, however, might increase the complexity of the optimization process, as the set of unitary operators that minimize the loss function is smaller compared to the case of using no entanglement.

\begin{figure}[t]
    \centering
    \includegraphics[width=1\linewidth]{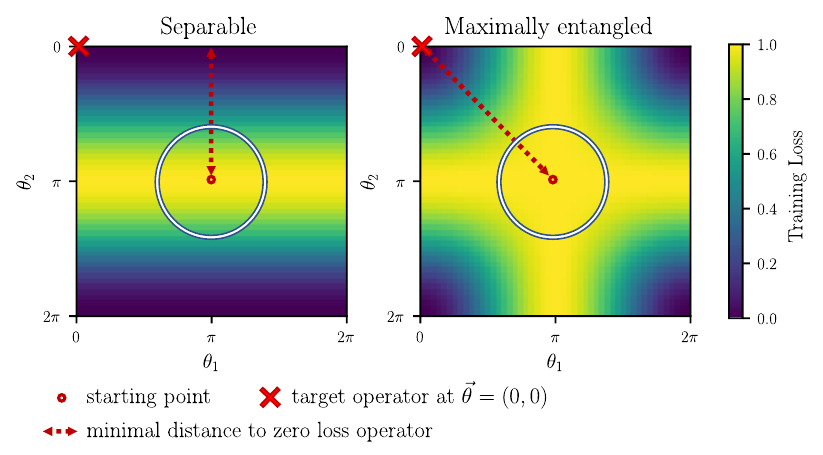}
    \caption{Comparison of loss landscapes for separable (left) and maximally entangled (right) training samples. The distance to the closest optimal solution, as determined by the used training sample, is shown by dotted arrows, with cross marks indicating the actual target operator.}
    \label{fig:small_pqc_cf}
\end{figure}

This fact is highlighted in the motivational example in \Cref{fig:small_pqc_cf}.
Herein, a two-parameter PQC $V(\theta_1, \theta_2)$, consisting only of a parameterized $X$- and $Z$-rotations on one qubit, is optimized.
The plots show the value of the loss function defined in \Cref{eq:loss_function} depending on the parameters $\vec{\theta}$ when a separable training sample (left plot) and when a maximally entangled training system (right plot) is used.
In this example, the operator $U=V(0,0)$ (marked by a cross) is the target operator for the supervised learning task.
The color at each point shows the loss function value, with darker colors indicating lower loss function values.
Starting from a randomly initialized parameter assignment for $\theta_1$ and $\theta_2$, the classical optimizer repeatedly evaluates the loss function to infer an optimization direction on this landscape.
However, for some parameter choices, the two shown loss landscapes differ, which might impact the performance of the optimizer.
As an example, consider the starting point $(\theta_1, \theta_2) = (\pi, \pi)$ highlighted in red in the center of each plot. 
For the separable training sample, the neighborhood delimited by a white circle around this point exhibits more variation in the loss function value than for the maximally entangled training sample.
Additionally, the distance to the closest global minimum on the loss landscape is larger for the maximally entangled training sample, as shown by the dotted arrows.

Together, these differences in the loss landscape for the maximally entangled training sample indicate an increase in the optimization complexity.
Thus, although using a maximally entangled state is beneficial for the risk after training, the performance of the training process might be impaired when compared to using a separable training sample.
Motivated by these observations, this work evaluates loss landscapes in supervised learning of unitary operators for larger problem sizes.
We begin by analytically describing the differences in the loss landscape between the maximally entangled state and the separable state in \Cref{sec:analytical} by assuming a maximally expressive training model.
Following the motivational example in \Cref{fig:small_pqc_cf}, we infer the variation of the loss function in a local neighborhood given by balls in a suitable metric and infer the distance to the closest global minimum on the loss landscape.
By comparing these values when a maximally entangled training sample is used and when an arbitrary separable training sample is used, we aim to infer if the possible negative effect on the optimization process in the motivational example is also apparent for larger problem instances.

In our analytical results, we assume that the model that is optimized is maximally expressive, i.e., it can explore the whole set of unitary operators to find suitable solutions to the problem.
However, as the expressivity of the used PQC might also influence the structure of the loss landscape~\cite{Holmes2022}, we augment our results by experimentally evaluating the properties of the loss landscape for a collection of PQCs and problem instances in \Cref{sec:exp}.
Following the motivational example, we simulate supervised learning of unitary operators in a constrained neighborhood. 
Thus, we analyze whether the decrease in variation of the loss function around randomly sampled starting points is also noticeable when training with a PQC as the supervised learning model.
Additionally, these experiments allow us to evaluate the effect of the degree of entanglement on the loss landscape by utilizing NME states for training.

\section{Optimizing over Unitary Operators}\label{sec:analytical}

\begin{figure*}[t!]
    \centering
        \includegraphics[scale=1]{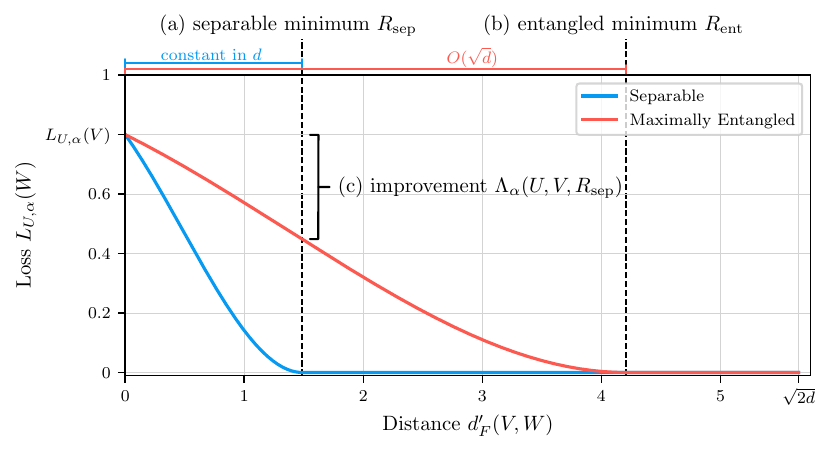}
    \caption{
    This plot summarizes the analytical results in \Cref{sec:analytical}. 
    It shows the minimal obtainable training loss for operators $W$ of distance $d_F'(V,W)$ from any operator $V$ with loss $L_{U,\arb}(V) = 0.8$ according to \Cref{thm:max_fid_theorem}.
    The minimal loss for a separable state is shown in blue, and the minimal loss for a maximally entangled state is shown in red. 
    The dashed vertical lines indicate the required distance until an operator $W$ of loss zero is found.
    For a separable state (a), this distance is constant in the dimension of the Hilbert space, while (b) it grows with $d = 2^n$ for the maximally entangled state in the worst case (\Cref{thm:min_distance_theorem}).
    Furthermore, the best possible improvement of the loss when a maximally entangled state is used is smaller than the improvement for the separable state. 
    As an example, consider the improvement for the maximally entangled training sample at (c) $d_F'(V,W) = R_{\text{sep}}$, which is shown in (\Cref{thm:same_start_value_theorem}) to be exponentially smaller than its counterpart in the separable case.}
    \label{fig:analytical_summary}
\end{figure*}

The loss function for supervised learning of unitary operators $U$, when a maximally entangled training sample is used, is closely related to the metric $d_F'(U, V)$ induced by the Frobenius norm after correcting for global phase factors (see \Cref{eq:frob_relationship_loss}).
By using this metric as a reference, we aim to evaluate the possible loss function values that are obtainable for separable training samples and compare them to the loss for maximally entangled training samples. 
Thus, we evaluate the behavior of the supervised learning loss when the set of possible hypothesis operators is assumed to be the whole set $\PU{d}$.
The main results of this evaluation are summarized in \Cref{fig:analytical_summary} and are explained in detail in the following.

\subsection{Distance to a Global Minimum}

Similar to the example in \Cref{sec:motivation}, we evaluate the distance to a global minimum on the loss landscape.
As the metric $d_F'$ can be defined using the fidelity $F_{U,\maxent}(V)$ for a maximally entangled state $\ket{\maxent} \in \cH_X \otimes \cH_R$ (\Cref{eq:frob_relationship_fidelity}), the distance between any hypothesis operator $V$ and the target operator $U$ during training directly follows from the loss at $V$.
Furthermore, since $d_F'$ is a metric on $\PU{d}$, any operator $V$ with $d_F'(U, V) = 0$ must match the target operator up to a global phase factor: $d_F'(U, V) = 0 \implies V=e^{i\rho}U$.
Thus, the loss function in supervised learning indicates the distance to any zero-risk operator $e^{i\rho}U$ when a maximally entangled training sample is used. 
However, when a separable training sample $\ket{\psi}$ is used, this relationship is not given. 
In particular, as shown in the motivation example, there might be other operators, apart from $e^{i\rho}U$, that constitute global minima for the loss function.

More generally, assuming that during the training process the algorithm produces a current hypothesis $V$, we proceed by describing the distance to any operator $W$ that satisfies a fixed target fidelity $f_W$ when a separable training sample is used.
For this task, we partition the set of possible hypothesis operators $\U{d}$ into distinct subsets $\Wset{\arb}(f_W)$ based on their fidelity.
Let $\ket{\arb}$ be an arbitrary training sample and let $f_W \in [0,1]$ be the target value for the fidelity, then 
\begin{align}
    \Wset{\arb}(f_W) := \Set{ W \in \U{d} |  F_{U,\arb}(W) = F((U^\dagger W \otimes I)\!\ket{\arb}, \ket{\arb}) = f_W }.\label{eq:def_W_set}
\end{align}
We characterize the minimal distance from the current solution $V$ to any $W \in \Wset{\psi}(f_W)$, using~$f_W$ and the fidelity of the current solution, denoted as $f_V$ as a shorthand.
\begin{restatable}{lemma}{repseparabledistance}
\label{le:separable_distance}
    Let $\ket{\psi} \in \cH_X \otimes \cH_R$ be a separable state with $\dim(\cH_X) = d$, let $U \in \U{d}$ be the target operator, and $V \in \U{d}$ the current hypothesis operator such that $F_{U,\psi}(V) = F((U^\dagger V \otimes I)\!\ket{\psi}, \ket{\psi}) = f_{V}$. For a fixed target fidelity $f_W \in [0,1]$, 
    \begin{align}
        \min_{W \in \Wset{\psi}(f_W)} d_F'(V, W) &= \sqrt{4\left(1 - \sqrt{f_V f_W} - \sqrt{(1-f_V)(1-f_W)}\right)}\\
        &= \sqrt{4 (1-\cos(\fsU{V}{\psi} - \fsU{W}{\psi}))}\label{eq:angle_reformulation},
    \end{align}
where $\fsU{V}{\psi} = \arccos(\sqrt{f_V})$ and $\fsU{W}{\psi} = \arccos(\sqrt{f_W})$.
\end{restatable}
\noindent The proof for this lemma is given in \Cref{app:cost_results}.

Setting the target value for the fidelity $f_W = 1$ in \Cref{le:separable_distance} allows us to infer the minimal distance to a global minimum of the loss function if a separable state is used for training:
\begin{align}
    \min_{W \in \Wset{\psi}(1)} d_F'(V, W) = \sqrt{4 (1-\sqrt{F_{U,\psi}(V)})} = \sqrt{4(1-\cos(\fsU{V}{\psi})}.\label{eq:mindf_sep}
\end{align}
According to the definition of $d_F'$ in \Cref{sec:metrics}, its maximal value $\sqrt{2d} = \sqrt{2 \cdot 2^n}$ grows exponentially with the number of qubits $n$.
However, the distance to any global maximum of the fidelity when a separable state is used for training stays independent of the dimension according to \Cref{eq:mindf_sep}.
This contrasts with the maximally entangled case, as the distance according to \Cref{eq:frob_relationship_fidelity} grows with the dimension $d$.
This suggests that although there is only a single global minimum in $\PU{d}$ in the maximally entangled case, there are numerous global minima in the separable case: \Cref{le:separable_distance} shows that one of these minima must lie within a constant distance of any point in the loss landscape.
However, the distance to these global minima still depends on the fidelity $F_{U,\psi}(V)$ and $F_{U,\maxent}(V)$ at the current solution $V$, which can differ between training samples.
Therefore, these two distances are not immediately comparable.
For this reason, we proceed by comparing the minimal distance to an optimum for separable and maximally entangled samples by distinguishing cases based on the loss at $V$.

\begin{restatable}{theorem}{repmindistancetheorem}
\label{thm:min_distance_theorem}
Let $\ket{\psi} \in \cH_X \otimes \cH_R$ be a separable and $\ket{\maxent} \in \cH_X \otimes \cH_R$ a maximally entangled training sample for $\dim(\cH_X) = d=2^n$ for $n$ qubits.
For the target operator $U \in \U{d}$ and the current hypothesis $V \in \U{d}$ define
\begin{align}
    W_{\psi} := \argmin_{W \in \Wset{\psi}(1)} d_F'(V, W)
\end{align}
and 
\begin{align}
    W_{\maxent} := \argmin_{W \in \Wset{\maxent}(1)} d_F'(V, W)
\end{align}
as the operators that are closest to the current hypothesis $V$ among those with maximal fidelity (and minimal loss) to the target operator $U$ for each sample. 
Then the following statements hold:
\begin{enumerate}[(i)]
    \abovedisplayskip=-\baselineskip
    \abovedisplayshortskip=-\baselineskip
        \item \begin{flalign}
        d_F'(V, W_{\psi}) \leq d_F'(V, W_{\maxent})
        &&
        \end{flalign}
    \abovedisplayskip=+\baselineskip
    \abovedisplayshortskip=+\baselineskip
        \item 
        If ${L_{U,\maxent}(V)} \in \lanom\left(\frac{1}{\mathrm{poly}(n)}\right)$, then
        \begin{align}
            \frac{d_F'(V, W_{\psi})}{d_F'(V, W_{\maxent})} \in \mathcal{O}\left(\frac{1}{\sqrt{d}} \right) = \mathcal{O}\left(\frac{1}{2^{n/2}}\right).
        \end{align}
    \item 
        If $\frac{d_F'(V, W_{\psi})}{d_F'(V, W_{\maxent})} \in \lanom\left( \frac{1}{\mathrm{poly}(n)}\right)$, then
        \begin{align}
            L_{U,\maxent}(V) \in \mathcal{O}\left(\frac{1}{2^n} \right).
        \end{align}
\end{enumerate}
\end{restatable}
\noindent The proof of this theorem is given in \Cref{app:cost_results}. 
Herein, we denote as $\mathrm{poly}(n)$ any polynomial in $n$ of finite degree.

The first statement emphasizes that the distance to a global minimum generally does not decrease if entangled training samples are used.
Thus, even if the loss at the current hypothesis $V$ is smaller for the maximally entangled training sample than for the separable training sample, the global minimum might still be more distant. 
The second statement concretizes this observation and is also indicated in \Cref{fig:analytical_summary}: as long as the loss for the maximally entangled sample is polynomially bounded from below in the number of qubits~$n$, the ratio between the distances to a minimum for the investigated training samples is exponentially small. 
Thus, a global minimum, when $\ket{\maxent}$ is used, is exponentially farther away.
The third statement evaluates the converse setting. 
If we find that the ratio of the distances to a global minimum is not exponentially small, herein expressed as polynomially bounded from below, then we can guarantee the current solution $V$ is already exponentially close to the target unitary $U$.
In other words, the only situation where we can ensure that the distance to a global minimum for $\ket{\maxent}$ is not exponentially large when compared to the distance to a global minimum for separable $\ket{\psi}$ is when $V$ is already exponentially close to $U$.

\subsection{Improvement in Local Neighborhood}\label{sec:improvement_analytical}

\Cref{le:separable_distance} describes the required distance to operators $W$ with a specific loss function value. 
This allows us to determine the distance $d_F'$ that needs to be covered by the optimizer to reach a hypothesis operator with a fixed loss function value.
However, it does not directly capture the properties of the immediate neighborhood of the current solution $V$, such as the variability of the loss function or the reachable minima in a neighborhood.
Thus, similar to the motivational example in \Cref{sec:motivation}, the following result specifies the properties of neighborhoods of an operator $V \in \U{d}$.
We describe these neighborhoods as balls with respect to the metric $d_F'$.

\begin{restatable}{lemma}{repmaxfidtheorem}\label{thm:max_fid_theorem}
    Let $B(V, R) := \Set{W \in \PU{d}|d_F'(V,W) \leq R}$ be a ball in $\PU{d}$ centered at the current hypothesis $V$ with radius $R$.
    For any separable state $\ket{\psi} \in \cH_X \otimes \cH_R$, the maximal fidelity to the target operator $U$ over all operators from the ball $B(V, R)$ is
    \begin{align}
        \max_{W\in B(V,R)} F((U^\dagger W \otimes I)\ket{\psi}, \ket{\psi}) = 
        \begin{cases}
            1 & \text{if } R \geq R_{\text{sep}}\\
             \cos\left(\fsU{V}{\psi} - \beta_{\text{sep}}\right)^2 & \text{otherwise}
        \end{cases},\label{eq:max_fid_sep}
    \end{align}
    with angles $\beta_{\text{sep}}= \arccos\left(1-\frac{R^2}{4}\right)$, $\fsU{V}{\psi} = \arccos\left(\sqrt{F_{U,\psi}(V)}\right)$ and threshold 
    \begin{align}
        R_{\text{sep}} = \sqrt{4}\sqrt{1-\sqrt{F_{U,\psi}(V)}}.
    \end{align}\\
    For any maximally entangled state $\ket{\maxent} \in \cH_X \otimes \cH_R$, the achievable maximal fidelity is
    \begin{align}
        \max_{W\in B(V,R)} F((U^\dagger W \otimes I)\ket{\maxent}, \ket{\maxent}) \leq 
        \begin{cases}
            1 & \text{if } R \geq R_{\text{ent}}\\
             \cos\left(\fsU{V}{\maxent} - \beta_{\text{ent}}\right)^2 & \text{otherwise}
        \end{cases},\label{eq:max_fid_ent}
    \end{align}
    with angles $\beta_{\text{ent}}= \arccos\left(1-\frac{R^2}{2d}\right)$, $\fsU{V}{\maxent} = \arccos\left(\sqrt{F_{U,\maxent}(V)}\right)$ and threshold 
    \begin{align}
        R_{\text{ent}} = \sqrt{2d}\sqrt{1-\sqrt{F_{U,\maxent}(V)}}.
    \end{align}
\end{restatable}
\noindent The proof of this lemma is given in \Cref{app:loss_improvement}.

This result specifies the best possible solution in a neighborhood around $V$ for the separable state, and it defines an upper bound for its equivalent for the maximally entangled case.
For both training samples, the maximal fidelity and, as a consequence, the minimal loss are governed by the radius $R$ of the ball and the fidelity at the current hypothesis $V$. 
Above the critical radii $R_{\text{sep}}$ and $R_{\text{ent}}$ respectively, the maximal fidelity is one, meaning an optimum for the supervised learning problem is contained in $B(V,R)$.
Below this threshold, the maximal fidelity depends on the difference of the Bures angle $\fsU{V}{\psi}$ at $V$ (or $\fsU{V}{\maxent}$ in the entangled case), and the angle $\beta_{\text{sep}}$ (or $\beta_{\text{ent}}$).
The latter angles quantify the largest possible deviation of Bures angles that are obtainable in $B(V,R)$.
They depend on the radius $R$ and on the used training sample and take values $\beta_{\text{sep}} \in \left[0,\arccos\left(\sqrt{F_{U,\psi}(V)}\right) = \fsU{V}{\psi}\right]$ for $0 \leq R \leq R_{\text{sep}}$ and $\beta_{\text{ent}} \in \left[0,\fsU{V}{\maxent}\right]$  for $0 \leq R \leq R_{\text{ent}}$.
In the entangled case, $\beta_{\text{ent}}$ additionally depends on the dimension of the state space.
As $d$ increases, $\beta_{\text{ent}}$ approaches zero.
Since $\beta_{\text{sep}}$ is independent of $d$, this suggests that the optimizer is able to obtain a larger improvement of the loss function value in a constrained neighborhood around $V$ when a separable training sample is used.

We evaluate this improvement in loss value by examining the largest decrease in loss function value in the ball $B(V,R)$ for $R \leq R_{\text{sep}}$ or $R \leq R_{\text{ent}}$, depending on the training sample.
The derivation of the bounds for the improvement is presented in \Cref{app:loss_improvement} and summarized in the following.
For an arbitrary training sample $\ket{\arb}$, we define the improvement $\imp_\alpha$ as the difference
\begin{align}\label{eq:improvement_def}
    \imp_{\arb}(U,V,R) := L_{U,\arb}(V) - \min_{W \in B(V,R)} L_{U,\arb}(W).
\end{align}

By using \Cref{thm:max_fid_theorem} and applying the angle sum identities for the resulting expression (see \Cref{eq:improvement_sin_sep} and \Cref{eq:improvement_sin_ent} in \Cref{app:loss_improvement}), the improvement of the loss function is
\begin{align}
    \imp_{x}(U,V,R) \leq \sin(2\gamma - \beta)\sin(\beta),\label{eq:improvement_reformulation}
\end{align}
for the states $\ket{x} \in \{\ket{\psi}, \ket{\maxent}\}$ with $\beta \in \{\beta_{\text{sep}}, \beta_{\text{ent}}\}$ and $\gamma \in \{\fsU{V}{\psi}, \fsU{V}{\maxent}\}$ as defined in \Cref{thm:max_fid_theorem}.
Furthermore, for the separable state $\ket{x} = \ket{\psi}$, this bound is an equality.
As $\sin(\beta_{\text{ent}}) \leq R/\sqrt{d}$, and since the dimension for the state space grows exponentially with the number of qubits, \Cref{eq:improvement_reformulation} implies that the obtainable improvement in $B(V,R)$ is exponentially small for maximally entangled states. 
We summarize this observation in the following result by considering the neighborhoods $B(V,R)$ around solutions $V$ with a fixed loss function value~$L$.

\begin{restatable}{theorem}{repsamestartvaluetheorem}\label{thm:same_start_value_theorem}
    Let $U \in \U{d}$ with $d=2^n$ be the target operator for the supervised learning problem and let $V_\psi$, $V_\maxent \in \mathcal{U}(d)$ be intermediate optimization solutions with $L_{U,\psi}(V_\psi) = L_{U,\maxent}(V_\maxent) = L$.
    For any radius $R \in \left(0, \sqrt{4(1-\sqrt{1-L})}\right]$,
    \begin{align}
        \frac{\imp_{\maxent}(U,V_\maxent,R)}{\imp_{\psi}(U,V_\psi,R)} \in \mathcal{O}\left( \frac{1}{2^{n/2}} \right).
    \end{align}
\end{restatable}
\noindent The proof of this theorem is given in \Cref{app:loss_improvement}.

Thus, if an optimization algorithm exhaustively examines constrained neighborhoods of radius $R$ around unitary operators that have the same loss $L$ with respect to the used training sample, it will obtain exponentially larger improvement in the separable case. 
We constrain the radius in the result above to $R \leq \sqrt{4(1-\sqrt{1-L})}$ as this is the radius required to obtain a solution of zero loss for the separable case according to \Cref{thm:max_fid_theorem}.
Thus, any neighborhood of larger radius will always allow for the maximally possible improvement for the separable training sample.
Of course, the assumption that the loss function values match for both training samples restricts the applicability of this result. 
However, the exponentially large denominator in the bound $\sin(\beta_{\text{ent}}) \leq R/\sqrt{d}$ suggests that the improvement for the maximally entangled training sample is exponentially small as long as this factor is not compensated for by the second factor in \Cref{eq:improvement_reformulation}.

\section{Evaluation of PQCs}\label{sec:exp}

The analytical results in \Cref{sec:analytical} show two effects of maximal entanglement in the training samples, for optimization in $\PU{d}$ with the Frobenius norm distance as metric:
\begin{itemize}
    \item The distance to a global optimum is 
    equal or larger when a maximally entangled state is used for training (\Cref{thm:min_distance_theorem}).
    \item The improvement of the loss in a neighborhood of fixed radius is smaller when a maximally entangled state is used (\Cref{thm:same_start_value_theorem}).
\end{itemize}
According to \Cref{eq:qnfl_theorem}, training with entangled training samples is beneficial for reducing the risk after training. 
Especially, a maximally entangled state such as $\ket{\Phi}$ implies a vanishing lower bound for the expected risk of the trained PQC. 
However, the observations above suggest that this improvement might come at the expense of a diminished training performance.

It is, however, not immediately clear whether these results translate to the case of training with specific PQCs.
On the one hand, PQCs are usually only able to explore a limited subset of the unitary group. 
Since our theoretical results rely on the existence of certain low-cost operators for the separable training sample (\Cref{le:separable_distance}), the applicability of these results is not immediately given.
On the other hand, the results rely on the closeness of certain operators with respect to the used metric $d_F'$ (\Cref{le:separable_distance}).
Whether this implies a closeness of operators on the loss landscape for PQCs is not clear.

Therefore, we experimentally evaluate the analytical results by performing constrained optimization with PQCs as the supervised learning model, using classical simulation of the quantum circuits.
The setup used for these experiments is presented in \Cref{sec:exp_setup}, and the evaluation of the experiment results is specified in \Cref{sec:exp_evaluation}.
The implementations and raw result data for the experiments are available online~\cite{Mandl2025datarepo}.

\subsection{Experiment Setup}\label{sec:exp_setup}
\begin{figure}
    \centering
    \includegraphics[scale=1]{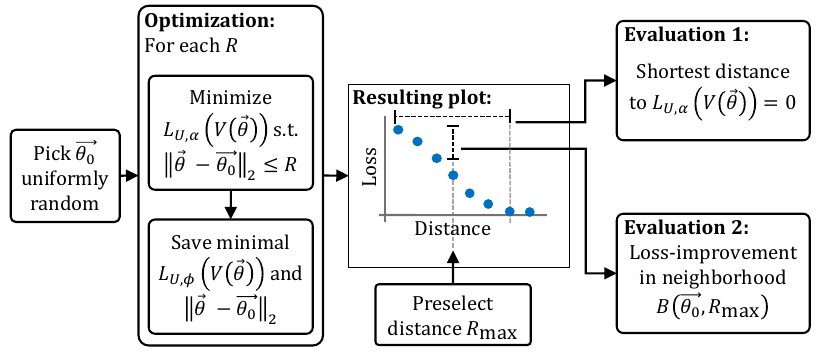}
    \caption{The experiment setup for evaluating the loss landscape for a training sample $\ket{\arb}$ around a randomly selected starting point $\vec{\theta_0}$. 
    For the selected PQC $V(\vec{\theta})$, we optimize the loss function $L_{U,\arb}(V(\vec{\theta}))$ while restricting the set of feasible parameter assignments $\vec{\theta}$ by a maximal 2-norm distance $R$.
    By iteratively increasing $R$ and saving the minimal loss and the associated distance for each run, we obtain a plot of minimal-loss points in constrained 2-norm balls around $\vec{\theta_0}$ (see also \Cref{fig:exp_output_example}). 
    From this plot, the minimal distance to an optimum (Evaluation 1) and the improvement of the loss function (Evaluation 2) up to a preselected distance $R_{\text{max}}$ (see \Cref{sec:exp_evaluation}) is inferred.}
    \label{fig:exp_setup}
\end{figure}
The general setup for our experiments uses a classical optimizer to minimize the supervised learning loss function in a constrained region of the loss landscape.
The setup and the evaluation of the results are visualized in \Cref{fig:exp_setup}.
For a given PQC $V : \mathbb{R}^p \to \mathcal{U}(d)$, we start by randomly sampling the starting point of the optimization $\vec{\theta_0} \in \mathbb{R}^p$.
We then optimize to find the minimum of the loss function when the set of feasible solutions is restricted to a ball $B(\vec{\theta_0}, R) = \Set{\vec{\theta} \; |\;\lVert \vec{\theta} - \vec{\theta_0} \rVert_2 \leq R }$ in the PQC's parameter space $\mathbb{R}^p$.
This procedure is repeated for increasingly large $R$ to obtain the distribution of minimal losses indicated in the schematic plot in \Cref{fig:exp_setup}.
We use this distribution to evaluate the shortest distance to any global minimum in the loss landscape (Evaluation 1) and the improvement within a fixed-distance neighborhood (Evaluation 2). 
The detailed descriptions of these evaluations, as well as the evaluated distances $R$, are presented in \Cref{sec:exp_evaluation}.

Our analytical results are based on the fact that for separable states, there are operators that obtain low training loss, while not fully matching the target operator.
Since for low expressivity it might not be given that a PQC can express these operators exactly, our analytical results might depend on the expressivity of the ansatz that is experimentally evaluated. 
Therefore, for our experiments, we aim to evaluate ansatzes of both high and low expressivity according to the metric described in \Cref{sec:pqcs}.

The expressivity of a broad selection of $19$ quantum circuits was evaluated in~\cite{Sim2019}. 
We therefore use a selection of the circuits employed therein for our experiments. 
In particular, we use two circuits that, in addition to parameterized rotation gates on each qubit, contain two-qubit gates applied to neighboring qubits to introduce entanglement (Circuit $4$ and Circuit $9$ in~\cite{Sim2019}).
Throughout this work, we will refer to these circuits after the two-qubit gates they use as \circfour and \circnine*.
We use these circuits as they cover a broad range of expressivity values: for a low number of circuit layers, they exhibit low expressivity, while for larger numbers of layers, the expressivity increases steadily~\cite{Sim2019}.
In contrast, we additionally use two PQC architectures that were shown in~\cite{Sim2019} to quickly saturate in expressivity as the number of layers increases.
On the one hand, we use Circuit $1$, which solely uses single-qubit gates, and we refer to it as \circone*.
On the other hand, we use Circuit $15$, which is comprised of a circular entanglement layer of $CX$ gates. 
We refer to this circuit as \circfifteen*.
For completeness, we present the circuit diagrams of each of the used PQCs in \Cref{app:expressivity}. 

For a finer-grained control over the expressivity of the employed PQCs, we perform all experiments for $l \in \{1,4,8,12,16\}$ ansatz layers and $n=5$ qubits.
This approach also leads to a varying dimension of the parameter space $\mathbb{R}^p$, ranging from $p = 5$ for \circnine with $l=1$ layer to $p = 176$ for \circfour with $l=16$ layers.

For our experiments, we aim to ensure that the possible inability of a PQC to express the target operator does not influence the results.
For this reason, we ensure that the PQC is always expressive enough for the problem that is to be solved, i.e., that zero loss after training is possible. 
To achieve this, we select the target operator by randomly sampling a parameter-assignment $\vec{\theta}_{\text{target}}$ and set $U=V(\vec{\theta}_{\text{target}})$.
To simulate the loss function calculation of the PQCs, we implement the PQCs in Pytorch~\cite{Paszke2019}.
To perform constrained optimization on the parameter space, we use the optimizer based on Sequential Least Squares Programming (SLSQP)~\cite{Kraft1988} implemented in Scipy~\cite{Virtanen2020_scipy}.
This optimizer allows for arbitrary constraints on the parameter space.
Thus, we limit the optimization process to the ball $B(\vec{\theta_0}, R)$ by specifying a maximal 2-norm distance to the starting point $\vec{\theta}_0$ as a constraint.

\subsection{Evaluation}\label{sec:exp_evaluation}
\begin{figure}
    \centering
    \includegraphics[scale=1]{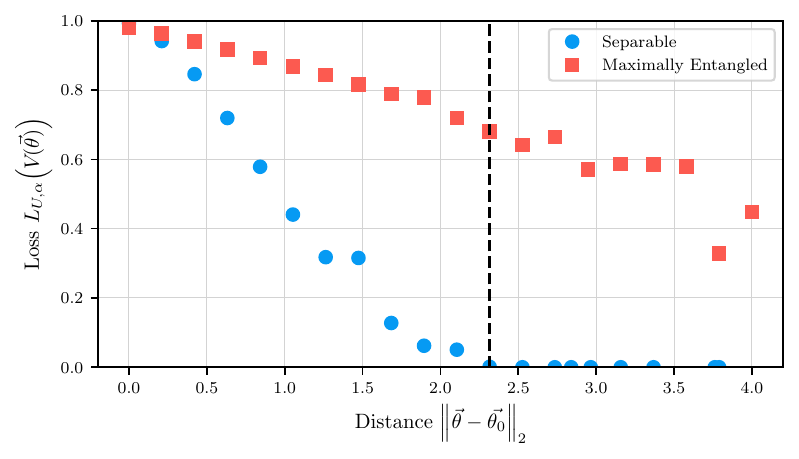}
    \caption{Minimal losses when training with a separable state (blue) and a maximally entangled state (red).
    By successively increasing the distance of feasible solutions, the optimizer finds increasingly better solutions. 
    At a distance of $\lVert \vec{\theta} - \vec{\theta_0} \rVert_2 \approx 2.3$, an optimum for the separable state is found (dashed line), while no optimum is yet found for the maximally entangled state.}
    \label{fig:exp_output_example}
\end{figure}

\Cref{fig:exp_output_example} shows the data obtained from our experiments for one example PQC $V(\vec{\theta})$.
The PQC is used to replicate an unknown operator $U$ starting from a randomly selected starting point $\vec{\theta_0}$.
The figure shows the decrease of the loss $L_{U,\arb}(V(\vec{\theta}))$ as the set of admissible solutions $B(\vec{\theta}, R)$ is increased. 
Thus, this plot shows a similar behavior to \Cref{fig:analytical_summary} for the analytical results in \Cref{sec:analytical}: the training loss decreases faster when a separable training sample (blue markers) is used. 
Furthermore, at a distance of $\lVert \vec{\theta} - \vec{\theta_0} \rVert_2 \approx 2.3$ (dashed line), a global minimum with zero loss is found for the separable training sample, while no such minimum is found yet for the maximally entangled training sample (red markers). 
This reinforces the analytical observation (\Cref{thm:min_distance_theorem}) that the distance to a minimum is generally smaller for separable training samples.
By calculating the difference between loss function values at distance zero and some distance $R$, the best possible improvement as defined in \Cref{eq:improvement_def} can be inferred.
For example, in \Cref{fig:exp_output_example} for $R = 1$, we find that the improvement for the separable state is considerably larger than for the maximally entangled state with improvements $\imp_\psi(U,V(\vec{\theta_0}), R=1) \approx 0.55$ and $\imp_\maxent(U,V(\vec{\theta_0}), R=1) \approx 0.1$.

In this section, we evaluate whether these observations hold for supervised learning of unitary operators in general.
For this task, we compare the obtained losses for the constrained optimization task across the previously described selection of PQCs in \Cref{sec:max_ent_experiment}.
This also allows us to study the influence of PQC properties, such as expressivity, on the improvement of the loss function. 
The exact values for the distance to a minimum (\Cref{thm:min_distance_theorem}) and the improvement of the loss function in a given radius $R$ (\Cref{eq:improvement_def}) depend on the starting loss $L_{U,\arb}(V(\vec{\theta_0}))$ and the used PQC itself. 
Furthermore, the approximated value for the improvement depends on the radius $R$ used for evaluation.
To obtain comparable evaluation results across runs, we therefore compare the following values:
\begin{itemize}
    \item The distance $R_{\text{max}}$ to a global minimum if such an optimum is found across the range of evaluated distances for all training samples.
    \item The improvement of the loss function at the distance where a global minimum is found for the separable training sample. Formally, for each training sample $\ket{\arb}$, we approximate the improvement (see \Cref{eq:improvement_def})
    \begin{align}
        \imp_{\arb}(U, V(\vec{\theta_0}), R_{\text{max}}) = L_{U,\arb}(V(\vec{\theta_0})) - \min_{\vec{\theta} \in B(\vec{\theta}, R_{\text{max}})} L_{U,\arb}(V(\vec{\theta})),\label{eq:improvement_experimental}
    \end{align}
    where $R_{\text{max}}$ is the smallest distance to a global minimum for the separable training sample in the current run.
\end{itemize}

In some cases, for example, for the maximally entangled training sample in \Cref{fig:exp_output_example} at distance $\lVert \vec{\theta} - \vec{\theta_0} \rVert_2 \approx4$, we observe an increase in the minimal obtained loss, although the evaluation region increased in radius.
These cases occur since we execute each run of the optimizer as an independent process and do not use previously obtained results as a fallback.
To mitigate the impact of these outliers on our results, we perform $24$ repetitions for each run of the experiment in \Cref{fig:exp_setup} and evaluate the overall distribution of the results in the following sections.

Lastly, as the analytical results in \Cref{sec:analytical} do not indicate how the training process is affected if NME states are used, we explore the possibility of using NME states for training in~\Cref{sec:nme_experiment}.

\subsubsection{Maximally Entangled States}\label{sec:max_ent_experiment}

\paragraph{Evaluation 1: Distance to Minimum}
\begin{figure}
    \centering
    \includegraphics[scale=1]{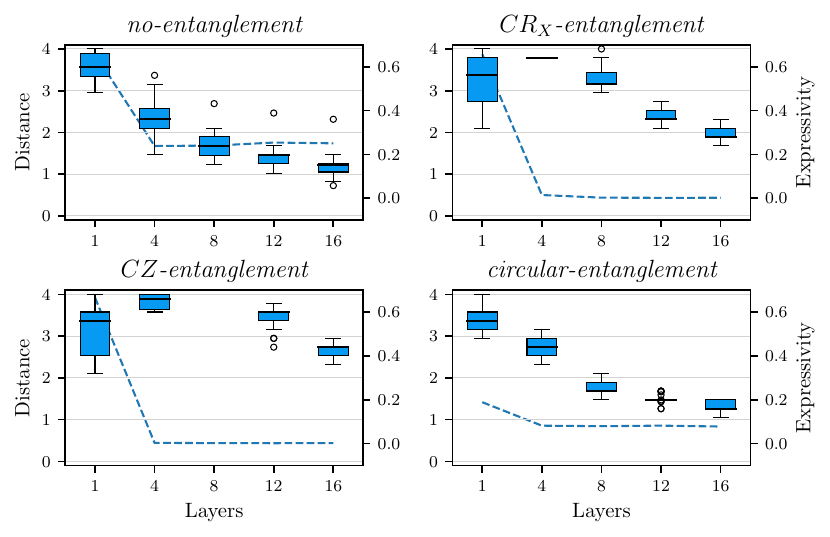}
    \caption{Distance to the closest minimum after optimizing in $B(\vec{\theta_0}, R) \leq 4$ for a separable training sample~$\ket{0}$.
    Cases where no minimum was found are excluded from the plot (e.g., \circnine*, $l=8$).
    The boxes highlight the median and extend from the first to the third quartile. The whiskers extend to 1.5 times the interquartile range.
    The secondary axis (dashed line) shows the expressivity (\Cref{sec:pqcs}) of the used PQCs, with lower values indicating high expressivity.}
    \label{fig:dist_to_zero}
\end{figure}

\Cref{fig:dist_to_zero} shows the 2-norm distance $\left\lVert \vec{\theta_0} - \vec{\theta}_{\text{opt}}\right\rVert_2$ from the starting point of the optimization to a global minimum when a separable state is used for training.
For the evaluation of the distance, we assume that a solution $V(\vec{\theta_{\text{opt}}})$ is a global minimum if $L_{U,\psi}(V(\vec{\theta}_{\text{opt}})) \leq 10^{-3}$.
For this setting, a global minimum for the separable state is consistently found in our maximal evaluation radius of $\lVert \vec{\theta} - \vec{\theta_0} \rVert_2 \leq 4$, regardless of the dimension of the parameter space of the PQC. 
As the number of layers of the used PQCs increases, so does the dimension $p$ of their parameter space $\mathbb{R}^p$.
In our experiments, we randomly sample both the starting point $\vec{\theta_0}$ as well as the target point $\vec{\theta}_{\text{target}}$ from this parameter space.
This would imply that the expected 2-norm distance between these points increases with the number of layers.
However, when comparing the results for $l=1$ and $l=16$ in \Cref{fig:dist_to_zero}, the plot suggests the opposite effect.
For the deep PQCs in our experiments, increasing the number of layers of the PQC decreases the distance required to find a minimum when a separable state is used for training.

In addition to the increase in parameter-space dimension, increasing the number of layers should also increase the expressivity of an ansatz.
To distinguish which of these effects causes the decrease in the distance to a minimum observed in our experiments, we further evaluate the expressivity.
\Cref{fig:dist_to_zero} shows the expressivity of each PQC on the secondary axis as a dashed line. 
The expressivity was evaluated according to the approach described in~\cite{Sim2019}, which is further elaborated in \Cref{app:expressivity}.
For $l\geq 4$ layers, the expressivity for \circfour and \circnine is maximal as indicated by a value close to zero. 
However, even though the expressivity does not significantly change by increasing the number of layers beyond $l\geq 4$, the distance to the minimum still decreases. 
This suggests that the decrease in the measured distances to the loss minimum is connected to the increased dimension $p$ of the parameter space.

\paragraph{Evaluation 2: Improvement in Local Neighborhood}
In the previous section, we evaluated the distance $R_{\text{max}}$  to a solution of zero loss for the separable training sample.
We proceed by comparing the improvement as defined in \Cref{eq:improvement_experimental} of the loss function in local neighborhoods around the starting point $\vec{\theta_0}$. 
Since the improvement for the separable state saturates at $\imp_{\psi}(U, V(\vec{\theta_0}), R_{\text{max}}) = L_{U,\psi}(V(\vec{\theta_0}))$, we define a local neighborhood to be the ball $B(\vec{\theta_0}, R_{\text{max}})$.
Thus, we evaluate the largest possible improvement for the separable state and the maximally entangled state in a neighborhood large enough to contain a minimum for the former.

\begin{figure}
    \centering
    \includegraphics[scale=1]{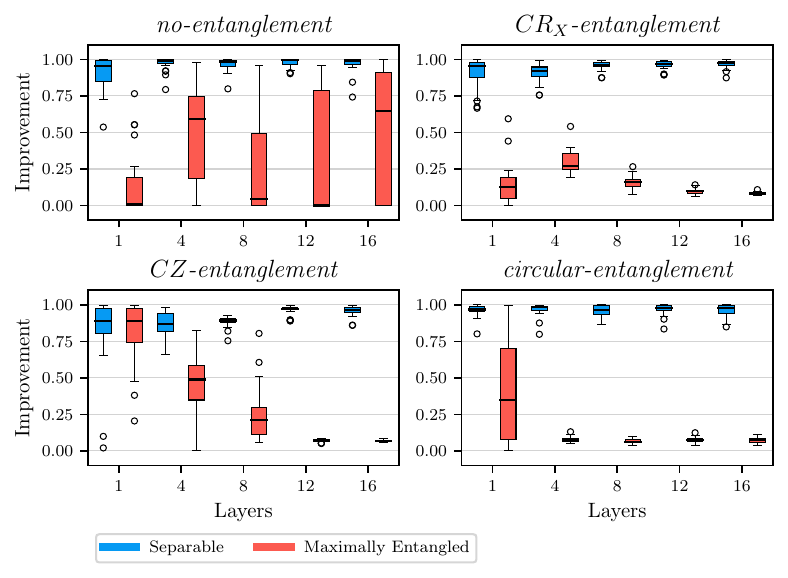}
    \caption{Improvement of the loss function in the ball $B(\vec{\theta_0}, R_{\text{max}})$ in parameter space. 
    The blue boxes show the improvement when a separable training sample $\ket{\psi}$ is used, and the red boxes show the improvement for a maximally entangled training sample $\ket{\maxent}$.}
    \label{fig:improvement_res}
\end{figure}

Randomly sampling the starting point for the optimization can lead to a high starting loss, as seen, for example, in \Cref{fig:exp_output_example}. 
Thus, the improvement for the loss function as shown in \Cref{fig:improvement_res} is consistently close to $1$ for the separable training sample as indicated by the blue boxes. 
One notable exception is \circnine*, where the improvement for $l=4$ and $l=8$ layers is slightly lower. 
This is due to the fact that for these instances, a solution with zero loss was not always found within $R\leq 4$.
However, even in these cases, the improvement of the loss function is mostly lower when a maximally entangled training sample is used (red boxes). 
Particularly when the number of layers increases, the difference between the improvement for the separable and the maximally entangled state increases.

A noteworthy exception to this observation is \circone*. 
Herein, the median improvement is largest for $l=16$.
Furthermore, the results for this PQC exhibit comparably large variance.
This follows in part from the setup of our experiments. 
As the target operator, we use $U=V(\vec{\theta}_{\text{target}})$ to ensure that the PQC is always able to express the target operator.
The PQC \circone is comprised solely of $R_X$ and $R_Z$ rotations without any two-qubit interactions. 
Since any single-qubit gate can be decomposed into $R_Z(\theta_1)R_X(\theta_2)R_Z(\theta_3)$ up to global phase differences~\cite{Barenco1995}, \circone is universal for the set of operators without two-qubit interactions as soon as $l\geq 2$.
Therefore, although this circuit has low expressivity, it is able to express any possible target operator in our experiments by adjusting a relatively low number of parameters.
Thus, even if the ansatz allows a large number of parameters, the target operator $U$ is found by only small adjustments of the parameters, even for the maximally entangled state.

\subsubsection{Non-Maximally Entangled States}\label{sec:nme_experiment}
\begin{figure}
    \centering
    \includegraphics[scale=1]{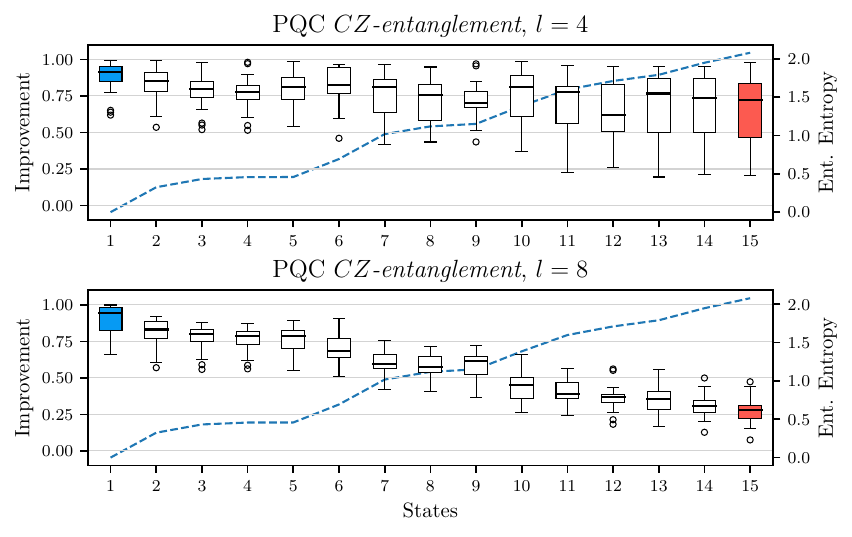}
    \caption{Improvement for NME states. Each box shows the evaluated improvement in a local neighborhood when training with an NME state. The states are ordered by their entanglement entropy, with low-entropy states on the left and high-entropy states on the right. The entanglement entropy is shown as a dashed line. The results for the separable state are highlighted in blue, and the results for the maximally entangled training sample are highlighted in red.}
    \label{fig:improvement_entanglement}
\end{figure}

The results in \Cref{sec:max_ent_experiment} show that for some PQCs, the improvement in a local neighborhood is decreased considerably if a maximally entangled state is used for training.
To further evaluate the influence of entanglement in training samples on supervised learning, we extend our experiments to include NME states.
In this section, we train \circnine with $l=4$ and $l=8$ layers using the approach described in \Cref{sec:exp_setup} for $n=3$ qubits.
As the training samples, we again use the separable state $\ket{\psi}$ and the maximally entangled state $\ket{\maxent}$.
However, we also include various states of intermediate entanglement by varying their Schmidt rank and Schmidt coefficients.

\Cref{fig:improvement_entanglement} shows the improvement obtained for these training samples. 
As the radius $R$ for the calculation of the improvement, we again use the smallest $R$ such that a minimum with zero loss is found for $\ket{\psi}$ and use the maximal value of $R=4$ if no such minimum is found.
Herein, each box corresponds to one training sample, and the results are ordered according to the entanglement entropy $E(\ket{\arb})$ (\Cref{eq:ent_entropy}) of each training sample $\ket{\arb}$, from lowest entropy on the left to highest entropy on the right.
The dashed line on the secondary axis shows the value of the entanglement entropy for each state.
Following the previous section, the separable training sample~$\ket{\psi}$ is highlighted in blue, and the maximally entangled training sample is highlighted in red.

For the larger PQC with $l=8$ layers, a dependence of the improvement on the entanglement entropy of the training sample is apparent. 
As the entanglement entropy increases from the separable state to the maximally entangled state, the improvement of the loss function in a fixed neighborhood decreases.
For $l=4$, this effect is less pronounced. 
However, the states of high entanglement entropy show a larger variance in their improvement value when compared to states of low entanglement entropy.
This indicates that it is still possible to obtain good improvement using these states in some cases, but it is increasingly unlikely.

\begin{figure}
    \centering
    \includegraphics[scale=1]{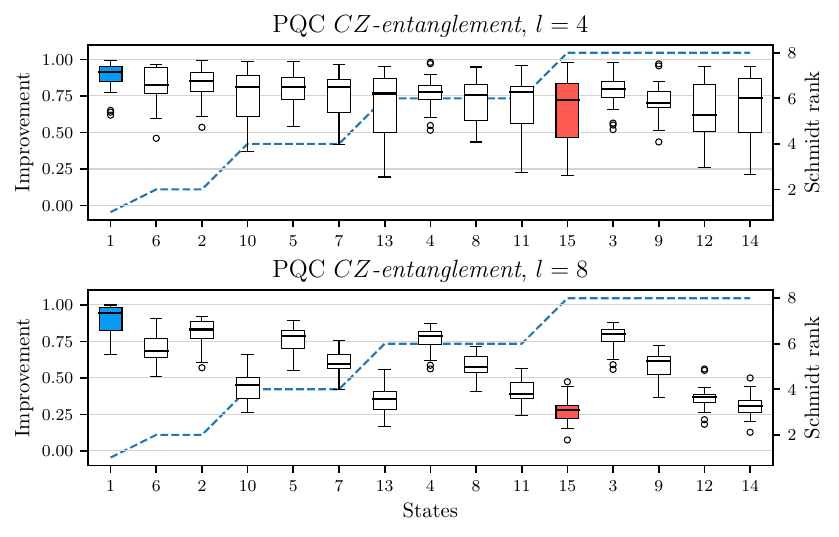}
    \caption{Improvement for NME states from \Cref{fig:improvement_entanglement} reordered based on the Schmidt rank of each training sample.
    The Schmidt rank is shown as a dashed line with values on the secondary axis.
    The state labels correspond to the labels in \Cref{fig:improvement_entanglement}.}
    \label{fig:improvement_schmidt_rank}
\end{figure}

Since the risk in supervised learning is affected by the Schmidt rank of the training samples, as opposed to the entanglement entropy, we proceed by evaluating our results based on the Schmidt rank of each sample. 
\Cref{fig:improvement_schmidt_rank} shows the improvement for NME states ordered by their Schmidt ranks, which is shown as a dashed line with values on the secondary axis.
The labels correspond to the labels of the states in \Cref{fig:improvement_entanglement} to aid comparison.
In contrast to \Cref{fig:improvement_entanglement}, we find no clear connection between the Schmidt rank and the improvement.
For example, for $l=8$, even for maximal Schmidt rank of $r=8$, there are states with improvement $\geq 0.5$.

\section{Discussion}\label{sec:discussion}

To evaluate the effect of entanglement in the training data on the supervised learning problem, we first simplified the problem setting in \Cref{sec:analytical} to training on the metric space described by the Frobenius norm distance. 
The results in this idealized setting were afterwards evaluated experimentally using actual PQCs.

In the simplified theoretical setting, we found that the distance to a global minimum in the loss landscape is exponentially larger when entanglement is used in the training samples.
This result is partly a consequence of existing results on the generalization error in supervised learning \Cref{sec:supervised_learning}.
The generalization error is expected to be minimal when a maximally entangled training sample is used and a zero loss solution is found during training.
Thus, the loss function for the maximally entangled sample can only be zero for solutions of zero generalization error. 
However, for the single separable training sample, this is not required: even if a solution of zero loss was reached during training, it is not guaranteed that this solution has minimal generalization error~\cite{Poland2020, Sharma2020}.
Consequently, there are more operators with zero loss for the separable state than there are for the maximally entangled state. 
\Cref{le:separable_distance} further shows that these operators are very common, as they are within a constant distance of any point on the loss landscape.

Unfortunately, although maximal entanglement in the training samples ensures that all optimization results have low generalization error, it introduces a drawback in the training complexity. 
The possible improvement of the loss function value (\Cref{sec:improvement_analytical}) in a constrained neighborhood scales inversely with the dimension of the state space. 
Thus, when learning operators with a large number of qubits, the exponential size of the state space implies that the improvement is exponentially small.
This entails a number of possible problems.
On the one hand, an exponentially small difference in loss function values implies that the gradient is exponentially small. 
Since classical optimizers often use approximations of this gradient, the precision of these approximations must increase with the dimension of the state space, which in turn increases optimization complexity~\cite{Thanasilp2023}.
On the other hand, these exponentially small variations in the loss function values could be overshadowed by noise on the quantum hardware.
This makes it increasingly hard to distinguish sampling noise from the gradient in the loss landscape, hindering the convergence of the training process.

In our analytical results, supervised learning algorithms are modeled as processes that explore the entire space $\PU{d}$ uniformly in search of a solution.
In implementations, this exploration is carried out by navigating the parameter space of a PQC.
However, this setting differs from the idealized one in two important ways: (i) a PQC generally cannot represent the complete set of unitary operators in $\PU{d}$, and (ii) the exploration of the accessible subset of unitaries is not guaranteed to be uniform.
Therefore, we evaluated our analytical findings using numerical simulations of PQC training.
We find that our results on the distance to a minimum in the optimization process (\Cref{le:separable_distance}) are well reflected in the experiments.
For the separable training sample, we consistently found a minimum in constant 2-norm distance.
Furthermore, we showed that this is not necessarily the case when the maximally entangled training sample is used.
In the worst case, the improvement of the loss function was significantly smaller than that achieved with the separable state. 

\subsection{PQC Expressivity}

To assess how the improvement of the loss function value depends on the expressivity of the PQC, we varied the number of layers of each ansatz.
The expressivity was measured by comparing the distribution of state fidelities obtained by randomly sampling from the PQC with that of states uniformly sampled according to the Haar measure~\cite{Sim2019}.
Increasing the number of layers increases the expressivity of an ansatz up to a certain limit.
\Cref{thm:min_distance_theorem} and \Cref{thm:same_start_value_theorem} assume that the model can express certain unitary operators of low loss for the separable state.
Thus, higher expressivity should entail a greater discrepancy in the improvement values for the training samples.
\Cref{fig:improvement_res} shows that this was generally the case when the number of layers was increased, with the exception of the PQC \circone*.
We ensured in our experiments that each used PQC is able to represent the solution to the learning problem.
Thus, a circuit of relatively low expressivity, such as \circone*, is still able to express the target operator.
Since the set of possible operators \circone can express is very limited, the chance of finding this target operator in relatively low distance, and thus showing large improvement, is increased.
Therefore, using circuits that only model a small subset of the set of unitary operators is preferable if it can be ensured that the circuit models the solution. 
This fact highlights that knowledge about the problem structure is beneficial in supervised learning.
For example, if it is known that the target operator is separable, using a PQC, such as \circone*, that solely models separable operators is sufficient.

Increasing the number of layers in a PQC not only increases expressivity but also increases the dimension of the parameter space. 
Thus, the increased complexity of the training process could also stem from the increased dimension of the parameter space.
However, since this increase in complexity would affect both training samples in our experiments equally, we argue that this is not the case.
Even with a large number of parameters, the optimizer consistently achieved a large improvement in the loss function value for the separable training sample. 
This reinforces the observation that the entanglement in the training data is detrimental to the optimizer's performance if highly expressive circuits are used.

\subsection{NME States}\label{sec:disc_nme}
Lastly, using our simulator experiments, we are able to evaluate the effect of entanglement when NME states are used. 
Especially for highly expressive circuits, we found that increasing the amount of entanglement in the training samples decreases the improvement in a neighborhood on the loss landscape. 
However, this decrease in improvement was not directly caused by the increase in the Schmidt rank of the training samples, but by the increase in entanglement entropy. 
For example, in \Cref{fig:improvement_schmidt_rank} there are states of high Schmidt rank, with large improvement ($\approx 0.75$ for state 3 using \circnine with $l=8$), and there are states with low Schmidt rank with comparably small improvement ($\approx 0.5$ for state 10 with Schmidt rank $4$).

The generalization error in supervised learning was shown to be linked to the Schmidt rank of the training sample, rather than the entanglement entropy.
This suggests that, theoretically, it is beneficial to use states of high Schmidt rank and low entanglement entropy for training.
The high Schmidt rank ensures that the generalization error stays low, while the low entanglement entropy ensures that the improvement during optimization is high.
However, low entanglement entropy implies that the Schmidt coefficients of these states are concentrated with high coefficients for one basis state and low coefficients for the others.
Thus, we suspect that training with such states might introduce local minima in the loss function where the trained operator has low loss on the high-coefficient basis states, but high loss on all others.
However, even if this is the case, NME states or separable states could serve as a warm start~\cite{Egger_2021, Truger2023} for the optimization process. 
Instead of only training with maximally entangled states, the optimization process could first train with states of low entropy and then fine-tune the optimization result using high-entropy entangled~states.

\section{Related Work}\label{sec:rel_work}
This work evaluates the training complexity of PQCs depending on the amount of entanglement of the training data with a reference system.
However, other sources of increased training complexity have also been identified in related work~\cite{Larocca2025}.
On the one hand, the expressivity of the PQC serves as an indicator of its susceptibility to barren plateaus, e.g., it was shown that the upper bound for the gradient variance is reduced by increasing the expressivity of the ansatz~\cite{Holmes2022}.
On the other hand, the observable that is used to infer the training loss influences the loss landscape.
Cerezo~et~al.~\cite{Cerezo2021a} show that for certain classes of PQCs, loss functions described by global observables, i.e., observables that act on all qubits, lead to barren plateaus.
Thus, for this setup, local loss functions are preferable.
By evaluating the effect of the training data on the loss function, these results were extended to capture the effect of training data on the loss landscape. 

Thanasilp~et~al.~\cite{Thanasilp2023} showed that even for local observables, entanglement of the training sample negatively affects the loss landscape. 
They evaluate ansatzes composed of tensor products of unitary operators acting on subsets of the available qubits. 
The entanglement of the training input, relative to one of these unitary operators, is quantified by tracing out all qubits that are not involved in the parameterized operator.
Their results show that the variance of the loss function, with respect to the parameters used in the parameterized operator, is proportional to the distance of the traced-out operator to the maximally mixed state~\cite{Thanasilp2023}. 
Thus, since a traced out maximally entangled state is maximally mixed, highly entangled states are detrimental to gradient variance in this setup.

Leone~et~al.~\cite{Leone2024} describe a similar effect for the Hardware Efficient Ansatz by using the entanglement entropy of the input states with respect to subdivisions of the input space that are given by the ansatz structure. 
By connecting this measure of entanglement to the amount of information that can be extracted from a subset of the input qubits, they show that a high entanglement entropy leads to concentration of the loss function values.
In the notation of our results, these works evaluate entanglement of the input states relative to subdivisions of the space $\cH_X$.
These subdivisions are specified by the structure of the ansatz that is used.
In contrast, we specifically evaluate the entanglement with a reference system and aim to keep the ansatz description general by considering optimization over all possible operators in $\PU{d}$. 
However, even with this difference in problem setup, our results lead to similar observations: a high entanglement entropy in the input states is detrimental to the optimization performance.

Starting from the various sources of barren plateaus that can arise in variational quantum algorithms, Cerezo~et~al.~\cite{Cerezo2023} investigate the question of whether barren plateaus are the consequence of the advantage that is obtained by employing quantum computers for these tasks. 
Quantum computers leverage the exponentially large operator spaces (e.g., unitary operators $V(\vec{\theta}) \in \U{d}$ in our case) for computation, but variational algorithms aim to classically optimize these operators using only a polynomial number of queries to the quantum computer.
They argue, by proving this connection for various applications, that the absence of barren plateaus implies that the algorithm only makes use of a polynomially large subspace of the Hilbert space and is thus also classically simulable~\cite{Cerezo2023}.
Thus, applied to the results in this work, this suggests that training with separable states could be classically simulable.
This reinforces the warm-starting approach proposed in \Cref{sec:disc_nme}: By starting the optimization process classically using training samples of low entanglement, initial parameters for an initialization for the optimization process with maximally entangled states can be obtained.
This initialization strategy might then prove useful in preventing problems during the optimization caused by barren plateaus.

Lastly, the training process requires that the fidelity can be accurately measured to infer the training loss.
Wang~et~al.~\cite{Wang2023} therefore evaluate the effect of the required measurement process on the risk after training when entanglement with a reference system is used and the observable is defined on $\cH_X$ only. 
They find that a limited number of measurements is detrimental to the risk after training when highly entangled training samples are used and provide guidelines on when entangled training samples are beneficial, based on the number of available measurements.
These results evaluate the performance by the risk after training, instead of by the loss landscape during training, as is done in this work.
However, they infer a similar connection between high entanglement and possibly decreased model performance.

\section{Conclusion}\label{sec:conclusion}
Although the introduction of entanglement with a reference system is beneficial for low-risk approximations in supervised learning, it might introduce additional complexity in the training process. 
In our analytical results, we show that this increase in complexity is reflected in a lower variation in loss function values within constrained neighborhoods of the loss landscape.
Furthermore, the generally lower variation in loss function values also causes an increase in the distance to a global minimum of the loss function when maximally entangled training samples are used.

We defined a neighborhood as balls given by the Frobenius norm distance as the metric. 
This metric does not necessarily have to capture closeness relations in the loss landscape when PQCs are used as the supervised learning model. 
Therefore, we extended the evaluation of the complexity of the loss landscape by experiments using the simulation of supervised learning using a selection of PQCs.
These experiments generally confirmed the analytical observation.
However, they also showed that the loss landscape is affected by the expressivity of the used quantum circuits. 
Highly expressive PQCs typically exhibited a larger decrease in loss function variation. 
Therefore, our results suggest that although high expressivity is generally preferred when no knowledge about the problem is available, this might lead to problems in the optimization if highly entangled training samples are used.

The experimental evaluation also allowed us to examine intermediate levels of entanglement. 
Our experiments show that NME states, even those of maximal Schmidt rank, might lead to a more favorable structure of the loss landscape than maximally entangled states.
This presents an opportunity to circumvent the problems introduced by maximal entanglement, either by using NME states throughout the whole training process or applying them in warm-starting procedures.
Therefore, we propose to explore this avenue in future work more thoroughly. 
Especially, it is unclear whether a pretraining approach that utilizes NME states could guide the classical optimizer to local minima. 
Furthermore, we used the Frobenius norm distance as the metric for the optimization space in our analytical results, because of its close connection to the loss function in the case of maximally entangled training samples.
However, the most appropriate metric for capturing the loss landscape depends on the specific PQC used. 
As such, an important direction for future work is to extend our analysis by evaluating the loss landscape under different metrics.

Lastly, as the measure of the expressivity in our experiments, we used the deviation of state fidelities for states generated by the PQC with state fidelities obtained by uniform sampling.
In our experiments on the distance to a minimum in PQC training, we observed that although the expressivity saturated for relatively low numbers of PQC layers, the distance still improved.
Thus, in this case, the used expressivity measurement could not fully predict the improvement in distance that was observed.
Since there are other ways of measuring expressivity presented in literature, we propose investigating whether the used approach influences these results.

\section*{Acknowledgements}
This work was partially funded by the BMWK projects \textit{EniQmA} (01MQ22007B) and \textit{SeQuenC} (01MQ22009B).

\appendix

\section{Distance to Fixed-Fidelity Operators}\label{app:cost_results}

In this section, we evaluate the Frobenius norm distance $d_F'(V,W)$ from a starting point $V\in\PU{d}$ to any operator $W$ with a fixed target fidelity $F_{U,\psi}(W) = f_W$. 
Thus, we find bounds on the distance to the closest operator in the set $\Wset{\psi}(f_W)$ for the separable state $\ket{\psi}$ defined in \Cref{sec:analytical}, and its counterpart $\Wset{\maxent}(f_W)$ for the maximally entangled state $\ket{\maxent}$.
In the first case, for $\ket{\psi}$, we provide a lower bound on the distance and show that this lower bound can be reached, i.e., there exists a minimal operator that attains the lower bound for the distance (\Cref{le:separable_distance}).
The distance to this operator is independent of the dimension $d$ of the Hilbert space.
In the latter case, for $\ket{\maxent}$, we provide a similar lower bound, which differs, however, in that it contains the dimension of the Hilbert space as a factor (\Cref{le:min_dist_entangled}).
This discrepancy leads to the observation that, in the worst case, a global minimum for the maximally entangled training sample is exponentially further from $V$ than a global minimum for the separable sample.
This observation is formally shown in \Cref{thm:min_distance_theorem}, which uses the following intermediate result.

\repseparabledistance*

\begin{proof}
    We use the relationship of the Frobenius norm distance and the absolute trace $\left| \Tr(V^\dagger W) \right|$ (\Cref{eq:rel_dist_trace}) and first provide an upper bound for the absolute trace for all $W \in \Wset{\psi}(f_W)$ and then proceed to construct a specific element in $\Wset{\psi}(f_W)$ that obtains this upper bound.
    For separable states, the reference system $\cH_R$ can be ignored. 
    Thus, for the remainder of this proof, we assume that $\ket{\psi} \in \cH_X$ refers to the factorization of the state $\ket{\psi}$ when the reference system is traced out.
    For the fidelity, we thus have $f_V= F_{U,\psi}(V) = \left| \braket{\psi|U^\dagger V|\psi}\right|^2$ and $f_W= F_{U,\psi}(W) = \left| \braket{\psi|U^\dagger W|\psi}\right|^2$.

    \emph{Upper bound:}
    For the proof of the upper bound, we express $\left| \Tr(V^\dagger W) \right|$ using any orthonormal basis (ONB) $\{\ket{\psi_j}\}_{j=1}^d$ with $\ket{\psi_1} = \ket{\psi}$:
    \begin{align}
        \left| \Tr(V^\dagger W) \right|
        &= \left| \braket{\psi_1 | V^\dagger W | \psi_1} + \sum_{j=2}^d \braket{\psi_j | V^\dagger W | \psi_j }\right|\\
        &\leq \left| \braket{\psi_1 | V^\dagger W | \psi_1}  \right| + \left|\sum_{j=2}^d \braket{\psi_j | V^\dagger W | \psi_j }\right|.\label{eq:trvwtriangle}
    \end{align}
    For the unitary operator $V^\dagger W$, existing results (see Theorem~1 in~\cite{Tromborg1978}) show that for every diagonal element $\braket{\psi_k | V^\dagger W | \psi_k}$, it holds that 
    \begin{align}
        \sum_{j=1}^d \left| \braket{\psi_j | V^\dagger W | \psi_j} \right| - 2\left| \braket{\psi_k | V^\dagger W | \psi_k} \right| \leq (d-2).
    \end{align}
    We rearrange this inequality for $k=1$ to obtain 
    \begin{align}
        \sum_{j=2}^d \left| \braket{\psi_j | V^\dagger W | \psi_j} \right| \leq (d-2) + \left| \braket{\psi_1 | V^\dagger W | \psi_1} \right|,
    \end{align}
    which, using \Cref{eq:trvwtriangle}, allows us to give an upper bound for the absolute trace
    \begin{align}
        \left| \Tr(V^\dagger W) \right| \leq 2\left| \braket{\psi_1 | V^\dagger W | \psi_1} \right| + (d-2).\label{eq:upper_bound_trace_first}
    \end{align}
    We proceed by expressing the remaining inner product in the equation above as 
    \begin{align}
        \braket{\psi_1 | V^\dagger W | \psi_1} = \braket{\psi | V^\dagger W | \psi} = \braket{\psi | V^\dagger U U^\dagger W| \psi }\label{eq:expansion_VW}
    \end{align}
    and write 
    \begin{align}
        U^\dagger V\!\ket{\psi} = \braket{\psi | U^\dagger V | \psi} \ket{\psi} + \ket{\psi_{U^\dagger V}^\bot}\label{eq:uv_bot}
    \end{align}
    and 
    \begin{align}
        U^\dagger W\!\ket{\psi} = \braket{\psi | U^\dagger W | \psi} \ket{\psi} + \ket{\psi_{U^\dagger W}^\bot}.\label{eq:uw_bot}
    \end{align}
    Herein, $\ket{\psi_{U^\dagger W}^\bot}$ and $\ket{\psi_{U^\dagger V}^\bot}$ are (not necessarily normalized) vectors with $\braket{\psi|\psi_{U^\dagger V}^\bot} = \braket{\psi|\psi_{U^\dagger W}^\bot} = 0$.
    By rewriting the inner product in \Cref{eq:expansion_VW} using these decompositions and applying the triangle inequality, we have 
    \begin{align}
        \left| \braket{\psi_1 | V^\dagger W | \psi_1}\right| &= \left| \braket{\psi| V^\dagger U U^\dagger W |\psi} \right|\\
        &= \left| \left( U^\dagger V\ket{\psi}\right)^\dagger \left(U^\dagger W \ket{\psi} \right) \right|\\
        &=  \left| \left( \braket{\psi | U^\dagger V | \psi} \ket{\psi} + \ket{\psi_{U^\dagger V}^\bot} \right)^\dagger \left( \braket{ \psi | U^\dagger W | \psi } \ket{ \psi } + \ket{ \psi_{U^\dagger W}^\bot } \right)  \right|\\
        &= \left| \left( \braket{ \psi | V^\dagger U | \psi } \bra{ \psi } + \bra{ \psi_{U^\dagger V}^\bot } \right) \left( \braket{ \psi | U^\dagger W | \psi } \ket{ \psi } + \ket{ \psi_{U^\dagger W}^\bot } \right)  \right|\\
        \begin{split}
        &= \Big| \braket{\psi|U^\dagger W |\psi} \braket{\psi|V^\dagger U |\psi} + \braket{\psi|U^\dagger W |\psi} \braket{\psi_{U^\dagger V}^\bot|\psi}\\
        &\qquad + \braket{\psi|V^\dagger U |\psi} \braket{\psi|\psi_{U^\dagger W}^\bot} + \braket{\psi_{U^\dagger V}^\bot|\psi_{U^\dagger W}^\bot}\Big|
        \end{split}\\
        &= \left| \braket{\psi|U^\dagger W |\psi} \braket{\psi|V^\dagger U |\psi} + \braket{\psi_{U^\dagger V}^\bot|\psi_{U^\dagger W}^\bot} \right|\\
        &\leq \left| \braket{\psi|U^\dagger W |\psi} \right| \left| \braket{\psi|U^\dagger V |\psi} \right| + \left| \braket{\psi_{U^\dagger V}^\bot|\psi_{U^\dagger W}^\bot} \right|\\
        &= \sqrt{f_V f_W} + \left| \braket{\psi_{U^\dagger V}^\bot|\psi_{U^\dagger W}^\bot} \right|.\label{eq:upperbound_first_step}
    \end{align}
    Using the Cauchy-Schwarz inequality, the rightmost summand is 
    \begin{align}
        \left| \braket{\psi_{U^\dagger V}^\bot|\psi_{U^\dagger W}^\bot} \right| \leq 
        \left\lVert \ket{\psi_{U^\dagger V}^\bot}\right\rVert \left\lVert\ket{\psi_{U^\dagger W}^\bot} \right\rVert\label{eq:csbotineq}.
    \end{align}
    For the right-hand side, we imply from the normalization of the state in \Cref{eq:uv_bot},
    \begin{align}
        1 &= \left\lVert \braket{\psi|U^\dagger V | \psi} \ket{\psi} + \ket{\psi_{U^\dagger V}^\bot}\right\rVert^2\\
        &= \left( \braket{\psi|U^\dagger V | \psi} \ket{\psi} + \ket{\psi_{U^\dagger V}^\bot} \right)^\dagger \left( \braket{\psi|U^\dagger V | \psi} \ket{\psi} + \ket{\psi_{U^\dagger V}^\bot} \right)\\
        &=  \left( \braket{\psi|V^\dagger U | \psi} \bra{\psi} + \bra{\psi_{U^\dagger V}^\bot} \right) \left( \braket{\psi|U^\dagger V | \psi} \ket{\psi} + \ket{\psi_{U^\dagger V}^\bot} \right)\\
        &= \braket{\psi|V^\dagger U | \psi}\braket{\psi|U^\dagger V | \psi} \braket{\psi|\psi} + \braket{\psi_{U^\dagger V}^\bot|\psi_{U^\dagger V}^\bot}\\
        &= \left|\braket{\psi|U^\dagger V|\psi}\right|^2 + \left\lVert\ket{\psi_{U^\dagger V}^\bot}\right\rVert^2.
    \end{align}
    Therefore,
    \begin{align}
        \left\lVert\ket{\psi_{U^\dagger V}^\bot}\right\rVert^2 
        &= 1 - \left|\braket{\psi|U^\dagger V|\psi}\right|^2\\
        &= 1 - f_V,\label{eq:1mfvform}
    \end{align}
    and using a similar argument,
    \begin{align}
        \left\lVert\ket{\psi_{U^\dagger W}^\bot} \right\rVert^2 &= 1-f_W,\label{eq:1mfwform}
    \end{align}
    follows from the norm of $U^\dagger W\!\ket{\psi}$ in \Cref{eq:uw_bot}.
    Therefore, by applying \Cref{eq:1mfvform} and \Cref{eq:1mfwform} to \Cref{eq:csbotineq} and \Cref{eq:upperbound_first_step}, we have
    \begin{align}
        \left| \braket{\psi_1 | V^\dagger W | \psi_1}\right| \leq \sqrt{f_V f_W} + \sqrt{(1 - f_V)(1 - f_W)}.
    \end{align}
    In conjunction with the upper bound in \Cref{eq:upper_bound_trace_first}, we can therefore give an upper bound for the absolute trace as 
    \begin{align}
        \left| \Tr(V^\dagger W) \right| \leq 2\left(\sqrt{f_V f_W} + \sqrt{(1 - f_V)(1 - f_W)}\right) + (d-2).\label{eq:trace_upper_bound}
    \end{align}

    \emph{Maximum:} 
    To show that the upper bound in \Cref{eq:trace_upper_bound} constitutes the maximum for $W \in \Wset{\psi}(f_W)$, we construct a specific unitary operator $\widetilde{W} = TV \in \Wset{\psi}(f_W)$ that obtains the upper bound.
    For this task, we define an ONB $\mathcal{B} = \left\{ \ket{b_1}, \dots, \ket{b_d} \right\}$  and use this basis to write $T$ in matrix form.
    Using this definition, we then show membership of $\widetilde{W}$ in $\Wset{\psi}(f_W)$ according to the definition in \Cref{eq:def_W_set} by inferring its unitarity and by calculating $F(U^\dagger \widetilde{W}\!\ket{\psi}, \ket{\psi})$.

    According to the decomposition in \Cref{eq:uv_bot}, there is a vector $\ket{\psi_{U^\dagger V}^\bot}$ that is orthogonal to $\ket{\psi}$.
    Let $\ket{\gamma}$ either be the normalized variant of this vector, i.e.~$\ket{\gamma} := \ket{\psi_{U^\dagger V}^\bot}/\lVert \ket{\psi_{U^\dagger V}^\bot} \rVert$, or let $\ket{\gamma}$ be any state with $\braket{\gamma|\psi}=0$ in case $\lVert \ket{\psi_{U^\dagger V}^\bot} \rVert = 0$.
    Using the Steinitz exchange theorem (see~\cite{Liesen2015}, 9.17) and the Gram-Schmidt procedure~\citep{Liesen2015}, we extend the set $\{U\!\ket{\psi}, U\!\ket{\gamma}\}$ by pairwise orthonormal vectors $\ket{b_3}, \dots, \ket{b_d}$ to obtain the ONB ${\mathcal{B} = \{\ket{b_1} := U\!\ket{\psi}, \ket{b_2} := U\!\ket{\gamma}, \ket{b_3}, \dots, \ket{b_d}\}}$.
    We define $T$ in the basis $\mathcal{B}$ as 
    \begin{align}
        T := \left(\begin{array}{@{}cc|c@{}}
            x & e^{i\theta} y & \\
            -e^{-i\theta}y & x & \smash{\raisebox{.5\normalbaselineskip}{$0$}}\\
            \hline
            &&\\[-0.5\normalbaselineskip]
            \multicolumn{2}{c|}{$0$} & I_{d-2}
            \\[-0.5\normalbaselineskip]
            &&
        \end{array}\right)\label{eq:def_T},
    \end{align}
    where $I_{d-2}$ is the identity on $\mathbb{C}^{d-2}$, $\theta := \arg\left(\braket{\psi|U^\dagger V|\psi}\right)$ is the phase angle of the inner product (with $\theta = 0$ if $\braket{\psi|U^\dagger V|\psi} = 0$), and 
    \begin{align}
        x:=\cos(\fsU{V}{\psi} - \fsU{W}{\psi}),\\
        y:=\sin(\fsU{V}{\psi} - \fsU{W}{\psi}),
    \end{align}
    The angles 
    \begin{align}
        \fsU{V}{\psi} = \arccos\left(\left| \braket{\psi|U^\dagger V|\psi} \right|\right) = \arccos\left(\sqrt{f_V}\right)\label{eq:repeated_def_gammaV}
    \end{align}
    and 
    \begin{align}
        \fsU{W}{\psi} = \arccos\left(\left| \braket{\psi|U^\dagger W|\psi} \right|\right) = \arccos\left(\sqrt{f_W}\right),
    \end{align}
    follow the definition of the Bures angle in \Cref{sec:metrics}.
    The matrix $T$ is unitary since 
    \begin{align}
        T^\dagger T &= \left(\begin{array}{@{}cc|c@{}}
            x & -e^{i\theta} y & \\
            e^{-i\theta}y & x &\smash{\raisebox{.5\normalbaselineskip}{$0$}}\\
            \hline
            &&\\[-0.5\normalbaselineskip]
            \multicolumn{2}{c|}{$0$} & I_{d-2}
            \\[-0.5\normalbaselineskip]
            &&
        \end{array}\right)
        \left(\begin{array}{@{}cc|c@{}}
            x & e^{i\theta} y & \\
            -e^{-i\theta}y & x & \smash{\raisebox{.5\normalbaselineskip}{$0$}}\\
            \hline
            &&\\[-0.5\normalbaselineskip]
            \multicolumn{2}{c|}{$0$} & I_{d-2}
            \\[-0.5\normalbaselineskip]
            &&
        \end{array}\right)\\
        &=
        \left(\begin{array}{@{}cc|c@{}}
            x^2 + y^2 & e^{i\theta} (xy - xy) & \\
            e^{-i\theta}(xy - xy) & x^2 + y^2 & \smash{\raisebox{.5\normalbaselineskip}{$0$}}\\
            \hline
            &&\\[-0.5\normalbaselineskip]
            \multicolumn{2}{c|}{$0$} & I_{d-2}
            \\[-0.5\normalbaselineskip]
            &&
        \end{array}\right)\\
        &=
        \left(\begin{array}{@{}cc|c@{}}
            1 & 0 & \\
            0 & 1 & \smash{\raisebox{.5\normalbaselineskip}{$0$}}\\
            \hline
            &&\\[-0.5\normalbaselineskip]
            \multicolumn{2}{c|}{$0$} & I_{d-2}
            \\[-0.5\normalbaselineskip]
            &&
        \end{array}\right),
    \end{align}
    where 
    \begin{align}
        x^2 + y^2 = \cos(\fsU{V}{\psi} - \fsU{W}{\psi})^2 + \sin(\fsU{V}{\psi} - \fsU{W}{\psi})^2 = 1.
    \end{align}
    Therefore, since $T$ and $V$ are unitary operators, $\widetilde{W} = TV \in \U{d}$.
    
   To calculate the fidelity $F(U^\dagger \widetilde{W}\! \ket{\psi}, \ket{\psi}) = F(U^\dagger TV\! \ket{\psi}, \ket{\psi})$, we first express the image of $\ket{\psi}$ after the application of $V$ in the basis $\mathcal{B}$.
    By rewriting \Cref{eq:uv_bot} using the definition of $\ket{\gamma}$ as
    \begin{align}
        U^\dagger V\!\ket{\psi} = \braket{\psi | U^\dagger V | \psi} \ket{\psi} + \left\lVert \ket{\psi_{U^\dagger V}^\bot} \right\rVert \ket{\gamma},
    \end{align}
    and multiplying from the left by $U$, it holds that
    \begin{align}
        V\!\ket{\psi} = \braket{\psi | U^\dagger V | \psi} U\!\ket{\psi} + \left\lVert \ket{\psi_{U^\dagger V}^\bot} \right\rVert U\!\ket{\gamma}.\label{eq:V_psi_rewrite}
    \end{align}
    Therefore, 
    \begin{align}
       F(U^\dagger TV\! \ket{\psi}, \ket{\psi}) &= \left|\braket{\psi|U^\dagger TV | \psi} \right|^2\\
       &= \left| \bra{\psi}U^\dagger \left( \braket{\psi | U^\dagger V | \psi} TU\!\ket{\psi} + \left\lVert \ket{\psi_{U^\dagger V}^\bot} \right\rVert TU\!\ket{\gamma}\right) \right|^2\\
       &= \left| \braket{\psi|U^\dagger V |\psi}\braket{\psi|U^\dagger T U|\psi} + \left\lVert \ket{\psi_{U^\dagger V}^\bot} \right\rVert \braket{\psi|U^\dagger T U |\gamma}\right|^2\\
       &= \left| \braket{\psi|U^\dagger V |\psi}\braket{b_1| T |b_1} + \left\lVert \ket{\psi_{U^\dagger V}^\bot} \right\rVert \braket{b_1| T |b_2}\right|^2\\
       &= \left| \braket{\psi|U^\dagger V |\psi}x + \left\lVert \ket{\psi_{U^\dagger V}^\bot} \right\rVert e^{i\theta} y\right|^2.
    \end{align}
    Herein, the last equality uses the definition of $T$ in the basis $\mathcal{B}$ in \Cref{eq:def_T} to substitute the matrix elements $\braket{b_1 | T | b_1} = T_{11}$ and $\braket{b_1 | T | b_2} = T_{12}$. Since $\theta$ is defined as the argument of $\braket{\psi|U^\dagger V|\psi}$, $\braket{\psi|U^\dagger V|\psi} = |\!\braket{\psi|U^\dagger V|\psi}\!| e^{i\theta}$.
    Thus, the arguments of both summands are equal, and the absolute value reduces to the sum of the magnitudes: 
    \begin{align}
        F(U^\dagger TV\! \ket{\psi}, \ket{\psi}) &= \left| \braket{\psi|U^\dagger V |\psi}x + \left\lVert \ket{\psi_{U^\dagger V}^\bot} \right\rVert e^{i\theta} y\right|^2\\
        &= \left| \left|\braket{\psi|U^\dagger V |\psi}\right| e^{i\theta} x + \left\lVert \ket{\psi_{U^\dagger V}^\bot} \right\rVert e^{i\theta} y\right|^2\\
        &= \left(\left|\braket{\psi|U^\dagger V |\psi}\right| x + \left\lVert \ket{\psi_{U^\dagger V}^\bot} \right\rVert y \right)^2\\
        &= \left(\cos(\fsU{V}{\psi})\;x + \left\lVert \ket{\psi_{U^\dagger V}^\bot} \right\rVert y \right)^2.
    \end{align}
    Herein, the last equality uses the definition of $\fsU{V}{\psi}$ in \Cref{eq:repeated_def_gammaV}. 
    From \Cref{eq:1mfvform}, we further have $\lVert \ket{\psi_{U^\dagger V}^\bot} \rVert = \sqrt{1-f_V} = \sqrt{1-\cos(\fsU{V}{\psi})^2} = \sin(\fsU{V}{\psi})$, therefore
    \begin{align}
        F(U^\dagger TV\! \ket{\psi}, \ket{\psi}) &= \left(\cos(\fsU{V}{\psi})x + \sin(\fsU{V}{\psi})y \right)^2\\
        &= \Big(\cos(\fsU{V}{\psi})\cos(\fsU{V}{\psi} - \fsU{W}{\psi}) \\
        &\qquad+ \sin(\fsU{V}{\psi})\sin(\fsU{V}{\psi} - \fsU{W}{\psi}) \Big)^2\\
        &= \cos(\fsU{V}{\psi} - \fsU{V}{\psi} + \fsU{W}{\psi})^2\label{eq:add_form_applied}\\
        &= \cos(\fsU{W}{\psi})^2\\
        &= f_W,
    \end{align}
    where we apply the trigonometric addition formula~\cite{Gelfand2001} in \Cref{eq:add_form_applied}.
    Therefore, since $\widetilde{W}$ is unitary and $F(U^\dagger \widetilde{W}\! \ket{\psi}, \ket{\psi}) = F(U^\dagger TV\! \ket{\psi}, \ket{\psi}) = f_W$, it follows that $\widetilde{W} \in \Wset{\psi}(f_W)$ according to the definition in \Cref{eq:def_W_set}.
    Lastly, to show that $\widetilde{W}$ achieves the upper bound from \Cref{eq:trace_upper_bound} by calculating $\left|\Tr(V^\dagger \widetilde{W})\right|$, we use the basis-invariance of the trace:
    \begin{align}
        \left| \Tr(V^\dagger \widetilde{W}) \right| &= \left| \Tr(V^\dagger TV) \right|\\
        &= \left| \Tr(T) \right|\\
        &= \left| 2x + \Tr(I_{d-2})\right|\\
        &= \left| 2\left( \cos(\fsU{V}{\psi} - \fsU{W}{\psi})\right) + (d-2)\right|\label{eq:trace_as_cos}\\
        &= \left|2\left( \cos(\fsU{V}{\psi})\cos(\fsU{W}{\psi}) + \sin(\fsU{V}{\psi})\sin(\fsU{W}{\psi})\right) + (d-2)\right|\label{eq:addition_theorem_expansion}\\
        &= 2\left(\sqrt{f_V f_W} + \sqrt{(1-f_V)(1-f_W)}\right) + (d-2),
    \end{align}
    Herein, \Cref{eq:addition_theorem_expansion} expands the result using the trigonometric addition formula~\cite{Gelfand2001} and the absolute value is omitted in the last equality since for $f_V, f_W \in [0,1]$, the expression is always nonnegative.
    Thus, it is possible to construct a unitary operator $\widetilde{W} \in \Wset{\psi}(f_W)$, that obtains the upper bound for $|\Tr(V^\dagger W)|$ in \Cref{eq:trace_upper_bound}.
    In other words,
    \begin{align}
        \max_{W \in \Wset{\psi}(f_W)} \left|\Tr(V^\dagger W)\right| = 2 \left( \sqrt{f_V f_W} + \sqrt{(1-f_V)(1-f_W)} \right) + (d-2).
    \end{align}
    By the relationship between the absolute trace and the Frobenius norm distance (\Cref{eq:rel_dist_trace}), the minimal distance is
    \begin{align}
        \min_{W \in \Wset{\psi}(f_W)} d_F'(V,W) &= \min_{W \in \Wset{\psi}(f_W)} \sqrt{2d \left( 1 - \frac{1}{d} \left|\Tr(V^\dagger W)\right| \right)}.\label{eq:rel_applied}
    \end{align}
    Since the square root is monotonically increasing on $\mathbb{R}_{\geq 0}$, it is minimal when its argument is minimal.
    For constant dimension $d$, this is the case in \Cref{eq:rel_applied} when the absolute trace $\left| \Tr(V^\dagger W)\right|$ is maximal.
    Therefore,
    \begin{align}
        \min_{W \in \Wset{\psi}(f_W)} d_F'(V,W) &= \sqrt{2d \left( 1 - \frac{1}{d} \max_{W \in \Wset{\psi}(f_W)} \left|\Tr(V^\dagger W)\right| \right)}\\
        &= \sqrt{2d \left( 1 - \frac{1}{d} \left(2 \left( \sqrt{f_V f_W} + \sqrt{(1-f_V)(1-f_W)} \right) + (d-2) \right)\right)}\\
        &= \sqrt{2 \left( d - 2 \left( \sqrt{f_V f_W} + \sqrt{(1-f_V)(1-f_W)}\right)  - d + 2\right)}\\
        &= \sqrt{4 \left( 1 - \left( \sqrt{f_V f_W} + \sqrt{(1-f_V)(1-f_W)}\right) \right)}\\
        &= \sqrt{4 \left( 1 -  \sqrt{f_V f_W} - \sqrt{(1-f_V)(1-f_W)}\right) }\label{eq:le1:before_last_reform}.
    \end{align}
    Lastly, using $\sqrt{f_V f_W} + \sqrt{(1-f_V)(1-f_W)} = \cos(\fsU{V}{\psi} - \fsU{W}{\psi})$ (compare \Cref{eq:trace_as_cos}), the maximum can be given in terms of the Bures angles $\fsU{V}{\psi}$ and $\fsU{W}{\psi}$ as
    \begin{align}
        \min_{W \in \Wset{\psi}(f_W)} d_F'(V,W) &=\sqrt{4 \left( 1 -  \cos(\fsU{V}{\psi} - \fsU{W}{\psi})\right) }.
    \end{align}
    
\end{proof}

The proof of \Cref{le:separable_distance} makes use of a decomposition of the states $U^\dagger V\!\ket{\psi}$ and $U^\dagger W\!\ket{\psi}$ into two orthogonal vectors, the input $\ket{\psi}$ and the orthogonal counterparts $\ket{\psi_{U^\dagger V}^\bot}$ and $\ket{\psi_{U^\dagger W}^\bot}$ (see \Cref{eq:uv_bot} and \Cref{eq:uw_bot}).
These orthogonal counterparts contain all errors introduced by the unitary operator when compared to the identity. 
By expressing their norms in terms of the fidelity, the upper bound for the distance $d_F'(V,W)$ is obtained.
This approach can be extended by decomposing the operators $U^\dagger V$ and $U^\dagger W$ themselves using the inner product $\braket{A,B}_F = \Tr(A^\dagger B)$, instead of decomposing the states after their application.
Using this decomposition, a similar bound is obtained for the distance between $V$ and $W$ given the target fidelity $f_W$ for the maximally entangled state.
\begin{lemma}\label{le:min_dist_entangled}
    Let $\ket{\maxent} \in \cH_X \otimes \cH_R$ be a maximally entangled state with $\dim(\cH_X) = d$, let $U \in \U{d}$ be the target operator, and $V \in \U{d}$ the current hypothesis operator such that $F_{U,\maxent}(V) = F((U^\dagger V \otimes I)\!\ket{\maxent}, \ket{\maxent}) = f_V$. 
    For a fixed target fidelity $f_W \in [0,1]$,
    \begin{align}
        \min_{W \in \Wset{\maxent}(f_W)} d_F'(V, W) &\geq \sqrt{2d\left(1 - \sqrt{f_V f_W} - \sqrt{(1-f_V)(1-f_W)}\right)}\\
        &= \sqrt{2d (1-\cos(\fsU{V}{\maxent} - \fsU{W}{\maxent}))}\label{eq:angle_reformulation},
    \end{align}
where $\fsU{V}{\maxent} = \arccos(\sqrt{f_V})$ and $\fsU{W}{\maxent} = \arccos(\sqrt{f_W})$.
\end{lemma}
\begin{proof}
    Similar to the proof of \Cref{le:separable_distance}, we provide an upper bound for $\left| \Tr(V^\dagger W)\right|$ and use this upper bound to give a lower bound for $d_F'(V,W)$.
    We start of by decomposing $U^\dagger V$ and $U^\dagger W$ into the identity and the operators $I_{U^\dagger V}^\bot$ and $I_{U^\dagger W}^\bot$ orthogonal to $I$ as
    \begin{align}
        U^\dagger V = \frac{\braket{I, U^\dagger V}_F}{d} I + I_{U^\dagger V}^\bot,\label{eq:mat_uv_bot}\\
        U^\dagger W = \frac{\braket{I, U^\dagger W}_F}{d} I + I_{U^\dagger W}^\bot\label{eq:mat_uw_bot}.
    \end{align}
    Herein, $\braket{I, I_{U^\dagger V}^\bot}_F = \Tr(I_{U^\dagger V}^\bot) = 0$ and 
    $\braket{I, I_{U^\dagger W}^\bot}_F = \Tr(I_{U^\dagger W}^\bot) = 0$.
    By rewriting $\left| \Tr(V^\dagger W)\right|$ in terms of this decomposition and applying the triangle inequality,
    \begin{align}
        \left| \Tr(V^\dagger W)\right| &= \left| \Tr(V^\dagger U U^\dagger W)\right| \\
        &= \left| \braket{U^\dagger V, U^\dagger W}_F\right| \\
        &= \left| \Braket{\frac{\braket{I, U^\dagger V}_F}{d} I + I_{U^\dagger V}^\bot, \frac{\braket{I, U^\dagger W}_F}{d} I + I_{U^\dagger W}^\bot }_F \right|\\
        \begin{split}
        &= \Biggl| \frac{\overline{\braket{I, U^\dagger V}}_F \braket{I, U^\dagger W}_F}{d^2} \braket{I,I}_F + \frac{\overline{\braket{I, U^\dagger V}}_F}{d}\braket{I, I_{U^\dagger W}}_F \\
        &\qquad + \frac{\braket{I, U^\dagger W}_F}{d}\braket{I_{U^\dagger V}^\bot, I}_F + \braket{I_{U^\dagger V}^\bot, I_{U^\dagger W}^\bot}_F\Biggr|
        \end{split}\\
        &= \left| \frac{\braket{U^\dagger V, I}_F \braket{I, U^\dagger W}_F}{d^2} d + \braket{I_{U^\dagger V}^\bot, I_{U^\dagger W}^\bot}_F\right|\\
        &\leq \frac{d}{d^2} \left| \braket{U^\dagger V, I}_F\right| \left| \braket{I, U^\dagger W}_F\right| + \left| \braket{I_{U^\dagger V}^\bot, I_{U^\dagger W}^\bot}_F \right|.
    \end{align}
    For the maximally entangled state, the fidelity $F_{U,\maxent}(V)$ is proportional to the absolute inner product (\Cref{eq:fidelity_max_ent_trace}):
    \begin{align}\label{eq:fidelity_trace_app}
        \frac{1}{d^2}\left|\braket{I, U^\dagger V}_F\right|^2 = \frac{1}{d^2}\left|\Tr(U^\dagger V)\right|^2 = F_{U,\maxent}(V).
    \end{align}
    Therefore, 
    \begin{align}\label{eq:proof_bound_trace}
        \left| \Tr(V^\dagger W)\right| \leq d \sqrt{f_V} \sqrt{f_W} + \left| \braket{I_{U^\dagger V}^\bot, I_{U^\dagger W}^\bot}_F \right|,
    \end{align}
    where we again used the shorthand $f_V = F_{U,\maxent}(V)$ and $f_W = F_{U,\maxent}(W)$ for the fidelity at $V$ and the target fidelity at $W$.
    Applying the Cauchy-Schwarz inequality for the rightmost summand shows 
    \begin{align}
        \left| \braket{I_{U^\dagger V}^\bot, I_{U^\dagger W}^\bot}_F \right| &\leq \left\lVert I_{U^\dagger V}^\bot \right\rVert_F \left\lVert I_{U^\dagger W}^\bot \right\rVert_F. \label{eq:I_CS_inequality}
    \end{align}
    Both norms in the expression above are derived from the decompositions in \Cref{eq:mat_uv_bot} and \Cref{eq:mat_uw_bot}, together with \Cref{eq:fidelity_trace_app} to write the norm in terms of the fidelity:
    \begin{align}
        d &= \braket{U^\dagger V, U^\dagger V}_F\\
        &= \Braket{\frac{\braket{I, U^\dagger V}_F}{d} I + I_{U^\dagger V}^\bot, \frac{\braket{I, U^\dagger V}_F}{d} I + I_{U^\dagger V}^\bot}_F\\
        \begin{split}
        &=  \frac{\overline{\braket{I, U^\dagger V}}_F \braket{I, U^\dagger V}_F}{d^2} \braket{I,I}_F + \frac{\overline{\braket{I, U^\dagger V}}_F}{d}\braket{I, I_{U^\dagger V}}_F \\
        &\qquad + \frac{\braket{I, U^\dagger V}_F}{d}\braket{I_{U^\dagger V}^\bot, I}_F + \braket{I_{U^\dagger V}^\bot, I_{U^\dagger V}^\bot}_F
        \end{split}\\
        &= \frac{d}{d^2} \left| \braket{I, U^\dagger V}_F\right|^2 + \braket{I_{U^\dagger V}^\bot,I_{U^\dagger V}^\bot}_F\\
        &= d\,f_V + \left\lVert I_{U^\dagger V}^\bot \right\rVert^2_F,
    \end{align}
    and as a result,
    \begin{align}
        \left\lVert I_{U^\dagger V}^\bot \right\rVert_F = \sqrt{d(1-f_V)}\label{eq:norm_IV}.
    \end{align}
    Similarly, 
    \begin{align}
        \left\lVert I_{U^\dagger W}^\bot \right\rVert_F = \sqrt{d(1-f_W)}\label{eq:norm_IW}.
    \end{align}
    Therefore, by using \Cref{eq:norm_IV} and \Cref{eq:norm_IW} in \Cref{eq:I_CS_inequality} in conjunction with \Cref{eq:proof_bound_trace}, an upper bound is derived
    \begin{align}
        \left| \Tr(V^\dagger W)\right| \leq d \sqrt{f_V f_W} + d \sqrt{(1-f_V)(1-f_W)}.
    \end{align}
    Lastly, since the upper bound for the absolute trace holds for all $W$ with $F_{U,\maxent}(W) = f_W$, we use this upper bound in \Cref{eq:rel_dist_trace} to obtain a lower bound for the sought-after distance:
    \begin{align}
        \min_{W \in \Wset{\maxent}(f_W)} d_F'(V, W) &= \min_{W \in \Wset{\maxent}(f_W)} \sqrt{2d \left(1 - \frac{1}{d} \left| \Tr(V^\dagger W) \right| \right)}\\
        &\sqrt{2d \left( 1 - \frac{1}{d} \max_{\Wset{\maxent}(f_W)} \left|\Tr(V^\dagger W)\right| \right)}\\
        &\geq 
        \sqrt{2d \left( 1 - \frac{d}{d} \left( \sqrt{f_V f_W} +  \sqrt{(1-f_V)(1-f_W)} \right) \right)}\\
        &= 
        \sqrt{2d \left( 1 -  \sqrt{f_V f_W} -  \sqrt{(1-f_V)(1-f_W)} \right)}.
    \end{align}
    Similar to \Cref{le:separable_distance}, this lower bound can be reformulated in terms of the angles $\fsU{V}{\maxent}$ and $\fsU{W}{\maxent}$ as 
    \begin{align}
        \min_{W \in \Wset{\maxent}(f_W)} d_F'(V, W) &\geq \sqrt{2d\left(1 - \cos(\fsU{V}{\maxent} - \fsU{W}{\maxent})\right)},
    \end{align}
    which concludes this proof.
\end{proof}

Using the target fidelity $f_W = 1$ for the globally optimal solution to the supervised learning problem with zero loss results in the distances 
\begin{align}
    \min_{\Wset{\psi}(1)} d_F'(V,W) = \sqrt{4\left(1-\sqrt{F_{U,\psi}(V)}\right)}
\end{align}
for the separable state and 
\begin{align}
    \min_{\Wset{\maxent}(1)} d_F'(V,W) \geq \sqrt{2d\left(1-\sqrt{F_{U,\maxent}(V)}\right)}\label{eq:distance_opt_ent_app}
\end{align}
for the maximally entangled state.
Since \Cref{eq:distance_opt_ent_app} depends on the exponentially large dimension $d$ of the state space, these bounds show that, at least when $F_{U,\psi}(V) = F_{U,\maxent}(V)$, the distance to a minimum is exponentially larger when $\ket{\maxent}$ is used.
In \Cref{thm:min_distance_theorem}, we extend this observation to show that this is always the case except when $V$ is already exponentially close to the target operator $U$.

\repmindistancetheorem*
\begin{proof}
    The first statement follows from the fact that any optimal $W_\maxent$, when $\ket{\maxent}$ is used, is equal to the target operator $U$ up to global phase differences. 
    Thus, since $W_{\maxent}$ matches $U$, the fidelity $F_{U,\psi}(W_{\maxent})$ must be maximal for any state $\ket{\psi}$.
    As a result, $W_{\maxent} \in \Wset{\psi}(1)$, which implies that the minimal distance to any $W \in \Wset{\psi}$ can not exceed the distance to $W_{\maxent}$.
    
    For statement (ii), any $W_\maxent$ has $d_F'(V, W_{\maxent}) = d_F'(U, V)$ since $W_\maxent$ matches $U$ apart from the global phase factor.
    According to \Cref{eq:frob_relationship_fidelity},
    \begin{align}
        d_F'(V, W_\maxent) = \sqrt{2d \left( 1- \cos(\fsU{V}{\maxent})\right)}.
    \end{align}
    Applying \Cref{le:separable_distance} and \Cref{le:min_dist_entangled} with $f_W = 1$, shows 
    \begin{align}
        \frac{d_F'(V, W_{\psi})}{d_F'(V, W_{\maxent})} &= \sqrt{\frac{4}{2d}} \sqrt{\frac{ 1- \cos(\fsU{V}{\psi})}{1- \cos(\fsU{V}{\maxent})}}\label{eq:fraction_equality}\\
        &\leq \sqrt{\frac{4}{2d}} \sqrt{\frac{ 1}{1- \cos(\fsU{V}{\maxent})}},
    \end{align} 
    where we use $1-\cos(\fsU{V}{\psi}) \leq 1$ for the Bures angle $\fsU{V}{\psi} \in [0, \frac{\pi}{2}]$.
    We further apply the lower bound $1-\cos(\fsU{V}{\maxent}) \geq \frac{4}{\pi^2} \fsU{V}{\maxent}^2$ for $\fsU{V}{\maxent} \in [0, \frac{\pi}{2}]$ (see Theorem~1 in~\cite{Bagul2022}), to obtain
    \begin{align}
        \frac{d_F'(V, W_{\psi})}{d_F'(V, W_{\maxent})} \leq \sqrt{\frac{4\pi^2}{2d}} \frac{1}{2\fsU{V}{\maxent}} = \sqrt{\frac{\pi^2}{2d}} \frac{1}{\fsU{V}{\maxent}} \label{eq:ratio_upper_bound}.
    \end{align}
    
    The assumption of statement (ii), $L_{U,\maxent}(V) = \sin(\fsU{V}{\maxent})^2 \in \lanom\left(\frac{1}{\mathrm{poly}(n)}\right)$, implies 
    \begin{align}
        \frac{1}{\sin(\fsU{V}{\maxent})^2} \in \mathcal{O}\left( \mathrm{poly}(n)\right).
    \end{align}
    Furthermore for $\fsU{V}{\maxent} \in [0, \pi/2]$, the lower bound $\fsU{V}{\maxent} \geq \sin(\fsU{V}{\maxent})$ applies~(see \cite{Qi2006}, 1.1).
    Therefore,
    \begin{align}
        \frac{1}{\fsU{V}{\maxent}^2} \leq \frac{1}{\sin(\fsU{V}{\maxent})^2} \in \mathcal{O}\left(\mathrm{poly}(n)\right)
    \end{align}
    and as a result
    \begin{align}
        \frac{1}{\fsU{V}{\maxent}} \in \mathcal{O}\left(\mathrm{poly}(n)\right).
    \end{align}
    Thus, since the reciprocal of the Bures angle $\fsU{V}{\maxent}$ is bounded above by a polynomial in $n$, the ratio in \Cref{eq:ratio_upper_bound} is dominated by the exponentially large dimension $d=2^n$ in its denominator for large $n$.
    Therefore, 
    \begin{align}
        \frac{d_F'(V, W_{\psi})}{d_F'(V, W_{\maxent})} \in \mathcal{O}\left( \frac{1}{\sqrt{d}}\right) = \mathcal{O}\left( \frac{1}{2^{n/2}}\right).
    \end{align}

    For statement (iii), we start at \Cref{eq:fraction_equality} and rewrite
    \begin{align}
        && \frac{d_F'(V, W_{\psi})}{d_F'(V, W_{\maxent})} &= \sqrt{\frac{4}{2d}} \sqrt{\frac{ 1- \cos(\fsU{V}{\psi})}{1- \cos(\fsU{V}{\maxent})}} \\
        \Leftrightarrow && \sqrt{1-\cos(\fsU{V}{\maxent})} &= \frac{d_F'(V, W_{\maxent})}{d_F'(V, W_{\psi})} \sqrt{\frac{4}{2d}} \sqrt{ 1- \cos(\fsU{V}{\psi})}
    \end{align}
    and apply $1-\cos(\fsU{V}{\psi}) \leq 1$ for $\fsU{V}{\psi} \in [0,\frac{\pi}{2}]$, and the lower bound $1-\cos(\fsU{V}{\maxent}) \geq \frac{4}{\pi^2} \fsU{V}{\maxent}^2$~\cite{Bagul2022},
    \begin{align}
        &&\frac{2}{\pi} \fsU{V}{\maxent} &\leq \frac{d_F'(V, W_{\maxent})}{d_F'(V, W_{\psi})} \sqrt{\frac{4}{2d}}\\
        \Leftrightarrow 
        && \fsU{V}{\maxent} &\leq \frac{\pi}{2} \frac{d_F'(V, W_{\maxent})}{d_F'(V, W_{\psi})} \sqrt{\frac{2}{d}}\label{eq:before_asympt}.
    \end{align}
    By the assumption of statement (iii), $\frac{d_F'(V, W_{\psi})}{d_F'(V, W_{\maxent})} \in \lanom\left( \frac{1}{\mathrm{poly}(n)}\right)$ and therefore,
    \begin{align}
        \frac{d_F'(V, W_{\maxent})}{d_F'(V, W_{\psi})} \in \mathcal{O}\left( \mathrm{poly}(n)\right)\label{eq:reformated_assuption2}.
    \end{align}
    Thus, the upper bound for the angle is dominated by the dimension $d=2^n$ in the denominator in \Cref{eq:before_asympt}, therefore $\fsU{V}{\psi} \in \mathcal{O}\left( \frac{1}{\sqrt{d}}\right)$.
    Lastly, since $\sin(\fsU{V}{\psi})^2 \leq \fsU{V}{\psi}$, it holds that 
    \begin{align}
        L_{U^\dagger V, \maxent} = \sin(\fsU{V}{\psi})^2 \in \mathcal{O}\left( \frac{1}{\sqrt{d}}\right) = \mathcal{O}\left( \frac{1}{2^{n/2}}\right).
    \end{align}

\end{proof}

\section{Loss Improvement}\label{app:loss_improvement}

We derive the training sample-dependent improvement of the loss function value in a neighborhood from bounds on the maximal fidelity in the neighborhood.
In the proofs in this section, we require upper and lower bounds for $\sin(x)$ for $x \in [0,\pi]$, which are derived here for completeness.
In~\cite{Bagul2022}, a refinement of Kober's inequality was shown:
\begin{align}
    \cos(x) \leq 1 - \frac{4x^2}{\pi^2},
\end{align}
for $x \in [0,\frac{\pi}{2}]$.
Since both sides of this inequality are symmetric with respect to the $y$-axis, this inequality also holds for $[-\frac{\pi}{2},\frac{\pi}{2}]$.
Using $\sin(x) = \cos\left(x-\frac{\pi}{2}\right)$, it follows that for $x \in [0,\pi]$,
\begin{align}
    \sin(x) &= \cos\left(x-\frac{\pi}{2}\right)\\
    &\leq 1-\frac{4}{\pi^2}\left(x-\frac{\pi}{2}\right)^2\\
    &=1-\frac{4}{\pi^2}\left(x^2 - x\pi + \frac{\pi^2}{4}\right)\\
    &= \frac{4x}{\pi}\left( 1 - \frac{x}{\pi}\right),\label{eq:sin_parabola}
\end{align}
which is a parabola that equals $\sin(x)$ at the endpoints of the interval, as well as at the maximum.
Furthermore, as a lower bound, we use the linear approximations
\begin{align}
    \sin(x) \geq \begin{cases}
        \frac{2}{\pi}x & 0\leq x \leq \frac{\pi}{2}\\
        2-\frac{2}{\pi}x & \frac{\pi}{2} \leq x \leq \pi. \label{eq:sin_piecewise}
    \end{cases}.
\end{align}
Herein, the first case is Jordan's inequality~\cite{Qi2006} and the second case follows since $\sin(x)$ is symmetric with respect to the axis at $\frac{\pi}{2}$.

The improvement of the loss function is derived from the following result.
\repmaxfidtheorem*

\begin{proof}
    We begin by deriving the maximal fidelity when a separable state is used (\Cref{eq:max_fid_sep}).
    Let $\fsU{V}{\psi} = \arccos\left(\sqrt{F_{U,\psi}(V)}\right)$,
    according to \Cref{le:separable_distance}, for each $f_W \in [0,1]$, there is a unitary operator $\widetilde{W}$ with $F_{U,\psi}(\widetilde{W}) = f_W$ and 
    \begin{align}
        d_F'(V,\widetilde{W}) = \sqrt{4\left(1-\cos\left(\fsU{V}{\psi} - \fsU{\widetilde{W}}{\psi}\right)\right)},\label{eq:min_distance_repeat}
    \end{align}
    where $\fsU{\widetilde{W}}{\psi} = \arccos\left(\sqrt{F_{U,\psi}(\widetilde{W})}\right)$.
    Furthermore, this operator is of minimal distance in the set of all operators with given $f_W$.
    Thus, if $R$ is larger than the distance in \Cref{eq:min_distance_repeat}, then the ball $B(V,R)$ includes an operator $W$ with $F_{U,\psi}(W) = f_W$.
    Therefore, we find the maximal $f_W \in [0,1]$, or equivalently, the minimal $\fsU{W}{\psi}$, such that \Cref{le:separable_distance} applies for the given radius $R$, i.e., for $R$ and $\fsU{V}{\psi}$ we find the minmal $\fsU{W}{\psi}$ that satisfies 
    \begin{align}
        &&\sqrt{4\left(1-\cos\left(\fsU{V}{\psi} - \fsU{W}{\psi}\right)\right)} &\leq R\\
        \Leftrightarrow&&\cos\left(\fsU{V}{\psi} - \fsU{W}{\psi}\right) &\geq 1- \frac{R^2}{4}\\
        \Leftrightarrow&&\fsU{V}{\psi} - \fsU{W}{\psi} &\leq \arccos\left(1- \frac{R^2}{4}\right) = \beta_{\text{sep}}.\label{eq:fvfw_condition}
    \end{align}

    For the proof of the first case in \Cref{eq:max_fid_sep}, we assume $R = R_{\text{sep}}$ and therefore
    \begin{align}
        \beta_{\text{sep}} &= \arccos\left(1-\frac{R_{\text{sep}}^2}{4}\right)\\
        &= \arccos\left(1-\left(1-\sqrt{F_{U,\psi}(V)}\right)\right)\\
        &= \fsU{V}{\psi}.
    \end{align}
    In this case, \Cref{eq:fvfw_condition} is satisfied for the minimally possible Bures angle $\fsU{W}{\psi}=0$ and for any $\fsU{V}{\psi}$ since 
    \begin{align}
        \fsU{V}{\psi} - \fsU{W}{\psi} = \fsU{V}{\psi} = \beta_{\text{sep}}. 
    \end{align}
    Thus, if $R=R_{\text{sep}}$, there is $W \in B(V,R)$ with $\fsU{W}{\psi}=0$ and as a result, $F_{U,\psi}(W) = 1$.
    Moreover, since $R \geq R_{\text{sep}}$ implies $B(V,R_\text{sep}) \subseteq B(V,R)$, the statement also follows for larger balls.

    In the second case, $R < R_{\text{sep}}$, assume $F_{U,\psi}(W) = \cos\left(\fsU{V}{\psi} - \beta_{\text{sep}} \right)^2$ as in the statement of the theorem.
    The corresponding Bures angle is 
    \begin{align}
        \fsU{W}{\psi} &= \arccos\left(\cos\left(\fsU{V}{\psi} - \beta_{\text{sep}} \right)\right)\\
        &= \fsU{V}{\psi} - \beta_{\text{sep}}.
    \end{align}
    Therefore,
    \begin{align}
        \fsU{V}{\psi} - \fsU{W}{\psi} = \beta_{\text{sep}},
    \end{align}
    which satisfies \Cref{eq:fvfw_condition}.
    Thus, there is an operator $W \in B(V,R)$ with         $F_{U,\psi}(W) = \cos\left(\fsU{V}{\psi} - \beta_{\text{sep}} \right)^2$.
    Furthermore, $W$ is minimal w.r.t. $\fsU{W}{\psi}$ since $\fsU{W}{\psi} < \fsU{V}{\psi} - \beta_{\text{sep}}$ implies $\fsU{V}{\psi} - \fsU{W}{\psi} > \beta_{\text{sep}}$, which contradicts \Cref{eq:fvfw_condition}.

    For the maximally entangled case (\Cref{eq:max_fid_ent}) we observe that $R_{\text{ent}} = d_F'(U,V)$ (\Cref{eq:frob_relationship_fidelity}). 
    Thus if $R \geq R_{\text{ent}}$, then $B(V,R)$ contains $U$ and the maximal fidelity is therefore $F_{U,\maxent}(U) = 1$.
    For the case $R < R_{\text{ent}}$ assume, for the sake of contradiction, that there is some $\widetilde{W} \in B(V,R)$ with $F_{U,\maxent}(\widetilde{W}) > \cos(\fsU{V}{\maxent} - \beta_{\text{ent}})^2$.
    By the definition of the Bures angle, 
    \begin{align}
        \fsU{\widetilde{W}}{\maxent} &= \arccos\left( \sqrt{F_{U,\maxent}(\widetilde{W})}\right)\\
        &< \arccos\left( \cos(\fsU{V}{\maxent} - \beta_{\text{ent}})\right)\\
        &= \fsU{V}{\maxent} - \beta_{\text{ent}}\label{eq:angle_beta_relationship}.
    \end{align}
    Herein, the inequality in valid since $\arccos(x)$ is decreasing in $[-1,1]$.
    \Cref{le:min_dist_entangled} allows us to infer the minimal distance to any operator from $\fsU{\widetilde{W}}{\maxent}$. 
    By using this minimal distance as a lower bound, it follows for $\widetilde{W}$ that
    \begin{align}
        d_F'(V, \widetilde{W}) &\geq \sqrt{2d\left(1 - \cos(\fsU{V}{\maxent} - \fsU{\widetilde{W}}{\maxent})\right)}.
    \end{align}
    \Cref{eq:angle_beta_relationship} implies $\beta_{\text{ent}} < \fsU{V}{\maxent} - \fsU{\widetilde{W}}{\maxent}$ and since $\beta_{\text{ent}} \geq 0$, it follows that both sides of this inequality are positive.
    Furthermore, both sides are bounded above by $\pi/2$.
    In this interval, $\cos(x)$ is a decreasing function, therefore $\cos(\beta_{\text{ent}}) > \cos(\fsU{V}{\maxent} - \fsU{\widetilde{W}}{\maxent})$.
    Thus the distance between $V$ and $\widetilde{W}$ is bounded by
    \begin{align}
        d_F'(V,\widetilde{W}) &> \sqrt{2d\left(1 - \cos(\beta_{\text{ent}})\right)}\\
        &= \sqrt{2d\left(1 - \cos\left( \arccos\left(1 - \frac{R^2}{2d} \right)\right)\right)}\\
        &= \sqrt{2d\frac{R^2}{2d}}\\
        &= R,
    \end{align}
    which contradicts the assumption that $\widetilde{W} \in B(V,R)$. 
    Therefore if $\widetilde{W} \in B(V,R)$, we have $F_{U,\maxent}(\widetilde{W}) \leq \cos(\fsU{V}{\maxent} - \beta_{\text{ent}})^2$.
\end{proof}

From the bounds in \Cref{thm:max_fid_theorem}, we can infer the improvement $\imp_\arb(U,V,R)$ of the loss function in a ball $B(V,R)$.
For any state $\ket{\arb} \in \cH_X \otimes \cH_R$, the improvement as defined in \Cref{eq:improvement_def} can be expressed in terms of the Bures angle $\fsU{V}{\arb}$ using \Cref{eq:loss_sinus} as 
\begin{align}
    \imp_{\arb}(U,V,R) &:= L_{U,\arb}(V) - \min_{W \in B(V,R)} L_{U,\arb}(W)\\
    &= \sin(\fsU{V}{\arb})^2 - \min_{W \in B(V,R)} L_{U,\arb}(W).
\end{align}
For the separable state $\ket{\psi}$, we focus on the case where $R \leq R_{\text{sep}}$, since otherwise the minimal loss is zero according to \Cref{thm:max_fid_theorem}, and the improvement always matches the initial loss at $V$. 
Since the minimal loss is the complement of the maximal fidelity, \Cref{thm:max_fid_theorem} shows 
\begin{align}
    \min_{W \in B(V,R)} L_{U,\psi}(W) &= 1 - \max_{W \in B(V,R)} F((U^\dagger W \otimes I)\ket{\psi}, \ket{\psi})\\
    &= 1 - \cos(\fsU{V}{\psi} - \beta_{\text{sep}})^2\\
    &= \sin(\fsU{V}{\psi} - \beta_{\text{sep}})^2.
\end{align}
Therefore,
\begin{align}
    \imp_{\psi}(U,V,R) &= \sin(\fsU{V}{\psi})^2 - \sin(\fsU{V}{\psi} - \beta_{\text{sep}})^2.\label{eq:imp_sep_app}
\end{align}
By factorizing this expression and applying the identities $\sin(x)\pm\sin(y) = 2\sin(\frac{x\pm y}{2})\cos(\frac{x\mp y}{2})$ and $\sin(x)\cos(x) = \frac{1}{2}\sin(2x)$~\cite{Gelfand2001}, it is simplified to 
\begin{align}
    \sin(\fsU{V}{\psi})^2 - \sin(\fsU{V}{\psi} - \beta_{\text{sep}})^2 &= \left(\sin(\fsU{V}{\psi}) - \sin(\fsU{V}{\psi} - \beta_{\text{sep}})\right)\\
    &\qquad\left(\sin(\fsU{V}{\psi}) + \sin(\fsU{V}{\psi} - \beta_{\text{sep}})\right)\\
    & = 4\sin\left(\frac{\beta_{\text{sep}}}{2}\right)\cos\left(\frac{2\fsU{V}{\psi} - \beta_{\text{sep}}}{2}\right)\\
    &\qquad\sin\left(\frac{2\fsU{V}{\psi} - \beta_{\text{sep}}}{2}\right)\cos\left(\frac{\beta_{\text{sep}}}{2}\right)\\
    &= \sin(2\fsU{V}{\psi} - \beta_{\text{sep}}) \sin(\beta_{\text{sep}}).\label{eq:improvement_sin_sep}
\end{align}
Using the upper bound for the maximal fidelity in the entangled case (\Cref{eq:max_fid_ent} in \Cref{thm:max_fid_theorem}), a similar expression is derived as an upper bound for the improvement for $\ket{\maxent}$:
\begin{align}
    \imp_{\maxent}(U,V,R) &= \sin(\fsU{V}{\maxent})^2  - \min_{W \in B(V,R)} L_{U,\maxent}(W)\\
    &\leq \sin(\fsU{V}{\maxent})^2 - \sin(\fsU{V}{\maxent} - \beta_{\text{ent}})^2\\
    &= \sin(2\fsU{V}{\maxent} - \beta_{\text{ent}}) \sin(\beta_{\text{ent}}).\label{eq:improvement_sin_ent}
\end{align}

Both expressions for the improvement (\Cref{eq:improvement_sin_sep} and \Cref{eq:improvement_sin_ent}) depend on the loss at $V$, in form of the Bures angles $\fsU{V}{\psi}$ and $\fsU{V}{\maxent}$ and on the respective angles $\beta_{\text{sep}}$ and $\beta_{\text{ent}}$.
While $\beta_{\text{sep}}$ depends solely on the radius $R$, $\beta_{\text{ent}}$ additionally depends on the dimension $d$ of the Hilbert space $\cH_X$.
Since $\sin(\arccos(x)) = \sqrt{1-\cos(\arccos(x))^2} = \sqrt{1-x^2}$, we have for the factor $\sin(\beta_{\text{ent}})$,
\begin{align}
    \sin(\beta_{\text{ent}}) &= \sin\left(\arccos\left(1-\frac{R^2}{2d}\right)\right)\\
    &= \sqrt{1 - \left(1-\frac{R^2}{2d}\right)^2}\\
    &= \sqrt{\frac{2R^2}{2d} - \frac{R^4}{4d^2}}\\
    &\leq \frac{R}{\sqrt{d}}.\label{eq:sin_beta_bound}
\end{align}

Since $d$ grows exponentially with the number of qubits in $\cH_X$, \Cref{eq:sin_beta_bound} indicates exponentially small improvement if a maximally entangled state is used for training.
However, \Cref{eq:improvement_sin_sep} and \Cref{eq:improvement_sin_ent} still depend on the Bures angles at the current solution $V$, we proceed by comparing the expressions for the improvement under the assumption that $\fsU{V}{\maxent} = \fsU{V}{\psi}$, i.e.~that the losses at $V$ match, in the following theorem.

\repsamestartvaluetheorem*
\begin{proof}
We first compare $\sin(\beta_{\text{sep}})$ and $\sin(\beta_{\text{ent}})$. 
From the assumptions of the theorem and the upper bound for the loss $L \leq 1$, $R \leq \sqrt{4} = 2$ follows, which implies $\frac{R^4}{16} \leq \frac{R^2}{4}$.
Thus, 
\begin{align}
    \sin(\beta_{\text{sep}}) &= \sin\left(\arccos\left(1-\frac{R^2}{4}\right)\right)\\
    &= \sqrt{\frac{2R^2}{4} - \frac{R^4}{16}}\\
    &\geq \sqrt{\frac{2R^2}{4} - \frac{R^2}{4}}\\
    &= \frac{R}{2}.
\end{align}
Therefore, using the upper bound in \Cref{eq:sin_beta_bound},
\begin{align}
    \sin(\beta_{\text{sep}}) \geq \frac{R}{2} = \frac{\sqrt{d}}{2} \frac{R}{\sqrt{d}} \geq \frac{\sqrt{d}}{2} \sin(\beta_{\text{ent}}).\label{eq:beta_relationship}
\end{align}
According to the assumptions of the theorem, $L_{U,\psi}(V_\psi) = L_{U,\maxent}(V_\maxent) = L$.
Thus, also the Bures angles $\fsU{V}{\maxent}$ and $ \fsU{V}{\psi}$ are equal and we denote this common angle as $\gamma := \arccos\left( \sqrt{1-L}\right)$.
By \Cref{eq:imp_sep_app}, \Cref{eq:improvement_sin_sep} and \Cref{eq:improvement_sin_ent} and using the relationship in \Cref{eq:beta_relationship},
\begin{align}
    \frac{\imp_{\maxent}(U,V_\maxent,R)}{\imp_{\psi}(U,V_\psi,R)} &\leq \frac{\sin(2\gamma - \beta_{\text{ent}})}{\sin(2\gamma - \beta_{\text{sep}})} \frac{\sin(\beta_{\text{ent}})}{\sin(\beta_{\text{sep}})}\\
    & \leq \frac{\sin(2\gamma - \beta_{\text{ent}})}{\sin(2\gamma - \beta_{\text{sep}})} \frac{2}{\sqrt{d}}.\label{eq:improvmenet_reform_1}
\end{align}
We proceed by upper-bounding the numerator and lower-bounding the denominator to get an upper bound for the ratio of the improvement.
From the assumption of the theorem, 
\begin{align}
    R \leq \sqrt{4(1 - \sqrt{1-L})} = \sqrt{4(1-\cos(\gamma))},\label{eq:R_cos_bound}
\end{align}
it follows that 
\begin{align}
    \cos(\gamma) = 1-\frac{4(1-\cos(\gamma))}{4} \leq  1-\frac{R^2}{4} \leq 1.
\end{align}
We apply this relationship to give the range for $\beta_{\text{sep}} := \arccos\left(1 - \frac{R^2}{4}  \right)$.
For $\gamma \in [0,\frac{\pi}{2}]$, it is $\cos(\gamma) \geq 0$.
Furthermore, since $\arccos(x)$ is a decreasing function, it is maximal at the lower end of the range of its argument, therefore
\begin{align}
    \beta_{\text{sep}} := \arccos\left(1 - \frac{R^2}{4}  \right) \in [0, \arccos(\cos(\gamma))] = [0, \gamma].
\end{align}
Therefore, the argument of the sine function in the denominator of \Cref{eq:improvmenet_reform_1} is constrained to the positive part of $\sin(x)$:
\begin{align}
    2\gamma - \beta_{\text{sep}} \in [2\gamma - \gamma, 2\gamma -0] = [\gamma, 2\gamma] \subseteq [0, \pi],\label{eq:betasep_bounds}
\end{align}
where the last inclusion follows from the range of the angle $\gamma \in [0, \pi/2]$.
For the entangled counterpart, the definition of $\beta_{\text{ent}}$ and the radius $R<2$ (see \Cref{eq:R_cos_bound}) show 
\begin{align}
    0 \leq \beta_{\text{ent}} = \arccos\left( 1-\frac{R^2}{4}\right)
    \leq \arccos\left( 1-\frac{R^2}{2d}\right) = \beta_{\text{sep}}
\end{align}
Thus, we infer $2\gamma - \beta_{\text{ent}} \in [0,\pi]$ similar to the separable case.

To upper bound the numerator in \Cref{eq:improvmenet_reform_1}, we use $\sin(x) \leq \frac{4x}{\pi}\left(1-\frac{x}{\pi}\right)$ (\Cref{eq:sin_parabola}):
\begin{align}
    \sin(2\gamma - \beta_{\text{ent}}) \leq \frac{4(2\gamma - \beta_{\text{ent}})}{\pi}\left(1-\frac{2\gamma - \beta_{\text{ent}}}{\pi} \right).
\end{align}
For the denominator in \Cref{eq:improvmenet_reform_1}, we use the piecewise approximation in \Cref{eq:sin_piecewise}:
\begin{align}
    \sin(2\gamma - \beta_{\text{sep}}) \geq \begin{cases}
        \frac{2}{\pi}\left(2\gamma - \beta_{\text{sep}}\right) & 0\leq 2\gamma - \beta_{\text{sep}} \leq \frac{\pi}{2}\\
        2-\frac{2}{\pi}\left(2\gamma - \beta_{\text{sep}}\right) & \frac{\pi}{2} \leq 2\gamma - \beta_{\text{sep}} \leq \pi. \label{eq:sin_piecewise}
    \end{cases}
\end{align}
To show the statement of the theorem, we examine both cases in \Cref{eq:sin_piecewise}. 
For $0\leq 2\gamma - \beta_{\text{sep}} \leq \frac{\pi}{2}$,
\begin{align}
    \frac{\imp_{\maxent}(U,V_\maxent,R)}{\imp_{\psi}(U,V_\psi,R)} &\leq \frac{\sin(2\gamma - \beta_{\text{ent}})}{\sin(2\gamma - \beta_{\text{sep}})} \frac{2}{\sqrt{d}}\\
    &\leq \frac{\frac{4(2\gamma - \beta_{\text{ent}})}{\pi}\left(1-\frac{2\gamma - \beta_{\text{ent}}}{\pi} \right)}{\frac{2}{\pi}\left(2\gamma - \beta_{\text{sep}}\right)}\frac{2}{\sqrt{d}}\\
    &\leq \frac{\frac{4\cdot 2\gamma}{\pi} (1-\frac{0}{\pi})}{\frac{2}{\pi}\left(2\gamma - \beta_{\text{sep}}\right)}\frac{2}{\sqrt{d}}\\
    & = \frac{4\cdot 2\gamma}{2\left(2\gamma - \beta_{\text{sep}}\right)}\frac{2}{\sqrt{d}}\\
    & \leq \frac{4\gamma}{\gamma}\frac{2}{\sqrt{d}}\\
    &= \frac{8}{\sqrt{d}} \in \mathcal{O}\left(\frac{1}{\sqrt{d}}\right).
\end{align}
Herein, the third inequality follows from $2\gamma - \beta_{\text{ent}} \geq 0$ and from $\beta_{\text{ent}} \geq 0$.
Similarly, the fourth inequality uses $2\gamma - \beta_{\text{sep}} \geq \gamma$ (\Cref{eq:betasep_bounds}).
For the second case in \Cref{eq:sin_piecewise},
\begin{align}
    \frac{\imp_{\maxent}(U,V_\maxent,R)}{\imp_{\psi}(U,V_\psi,R)} &\leq \frac{\frac{4(2\gamma - \beta_{\text{ent}})}{\pi}\left(1-\frac{2\gamma - \beta_{\text{ent}}}{\pi} \right)}{2-\frac{2}{\pi}\left(2\gamma - \beta_{\text{sep}}\right)}\frac{2}{\sqrt{d}}\\
    &\leq \frac{\frac{4\pi}{\pi}\left(1-\frac{1}{\pi}\left(2\gamma - \beta_{\text{ent}}\right) \right)}{2\left(1-\frac{1}{\pi}\left(2\gamma - \beta_{\text{sep}}\right)\right)}\frac{2}{\sqrt{d}}\label{eq:second_case_1}\\
    &= \frac{1-\frac{2\gamma}{\pi} + \frac{\beta_{\text{ent}}}{\pi}}{1-\frac{2\gamma}{\pi} + \frac{\beta_{\text{sep}}}{\pi}}\frac{4}{\sqrt{d}}\\
    &\leq\frac{1-\frac{2\gamma}{\pi} + \frac{\beta_{\text{sep}}}{\pi}}{1-\frac{2\gamma}{\pi} + \frac{\beta_{\text{sep}}}{\pi}}\frac{4}{\sqrt{d}}\label{eq:second_case_2}\\
    & = \frac{4}{\sqrt{d}}\in \mathcal{O}\left(\frac{1}{\sqrt{d}}\right).
\end{align}
In \Cref{eq:second_case_1}, we apply $2\gamma - \beta_{\text{ent}} \leq 2\gamma \leq \pi$ and in \Cref{eq:second_case_2}, we make use of $\beta_{\text{ent}} \leq \beta_{\text{sep}}$.
Since $d=2^n$, 
\begin{align}
    \frac{\imp_{\maxent}(U,V_\maxent,R)}{\imp_{\psi}(U,V_\psi,R)} \in \mathcal{O}\left(\frac{1}{2^{n/2}}\right)
\end{align}
in both cases.
\end{proof}

\section{Used PQCs and their Expressivity}\label{app:expressivity}
\begin{figure}
    \centering
        \begin{subfigure}[b]{0.375\textwidth}
            \centering
            \includegraphics[scale=0.8]{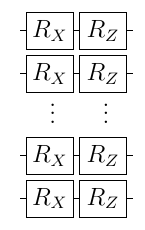}
            \caption{no-entanglement}   
            \label{fig:diag_no_entanglement}
        \end{subfigure}
        \hfill
        \begin{subfigure}[b]{0.575\textwidth}  
            \centering 
            \includegraphics[scale=0.8]{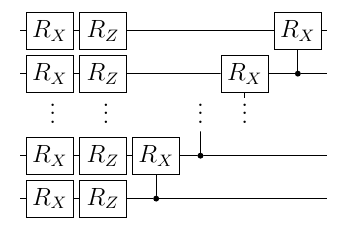}
            \caption{$CR_X$-entanglement}   
            \label{fig:diag_crx_entanglement}
        \end{subfigure}
        \vskip\baselineskip
        \begin{subfigure}[b]{0.375\textwidth}   
            \centering 
            \includegraphics[scale=0.8]{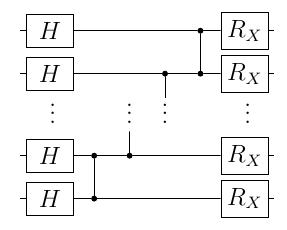}
            \caption{$CZ$-entanglement}   
            \label{fig:diag_cz_entanglement}
        \end{subfigure}
        \hfill
        \begin{subfigure}[b]{0.575\textwidth}   
            \centering 
            \includegraphics[scale=0.8]{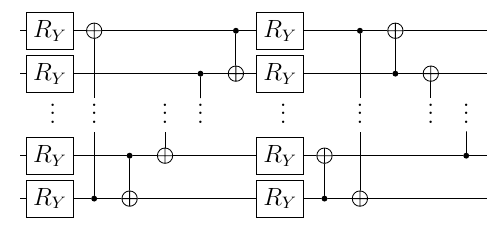}
            \caption{circular-entanglement}   
            \label{fig:diag_circ_entanglement}
        \end{subfigure}
        \caption{Circuit diagrams of the PQCs evaluated in this work. Each contained rotation gate ($R_X, R_Y, R_Z$) is parameterized, and the parameters are used for optimization.
        The circuit labels in this figure are used in the main text as a shorthand to refer to the PQCs.} 
        \label{fig:circ_diagrams}
\end{figure}

For our experiments in \Cref{sec:exp}, we utilize a subset of the circuits evaluated in~\cite{Sim2019}, selected to cover a range of entanglement structures, including circuits without entanglement, with $CR_X$ entanglement, $CZ$ entanglement, and circular entanglement.
\Cref{fig:circ_diagrams} shows the circuit diagram of each PQC, along with the label that we use in the main text to refer to the PQC.
To evaluate and compare the expressivity in our experiments, we follow the approach proposed in~\cite{Sim2019}.
The expressivity of a PQC $V(\vec{\theta})$ is experimentally estimated by comparing the probability distributions of the fidelity of Haar random states with the distribution of fidelities of states obtained from randomly selecting PQC parameters. 
Formally, it is estimated as the Kullback-Leibler divergence 
\begin{align}
    \mathrm{Expr} := D_{KL}\left( \hat{P}_V (F, \vec{\theta})\;\lVert\; P_{\text{Haar}}(F)\right)
\end{align}
of the estimated probability distribution $\hat{P}_V(F, \vec{\theta})$ of fidelities 
\begin{align}
    F=\left|\braket{0|V(\vec{\theta})^\dagger V(\vec{\rho}) | 0}\right|^2,\label{eq:pqc_fid}
\end{align}
when $\vec{\theta}$ and $\vec{\rho}$ are sampled uniformly at random and the probability distribution $P_{\text{Haar}}(F)$ of fidelities $F=\left|\braket{\psi|\phi}\right|^2$ of Haar-random states $\ket{\psi}$ and $\ket{\phi}$~\cite{Sim2019,Nakaji2021}.

To estimate the expressivity for the circuits used in the experiments, we sample $N$ pairs of parameter assignments $\vec{\theta}$ and $\vec{\rho}$ and compute the fidelity according to \Cref{eq:pqc_fid}.
We classify the results into $N_{\text{bins}}=75$ equally sized bins $[a_i,b_i) \subset [0,1]$, with the last bin closed to cover the full interval $[0,1]$, to obtain a histogram for the distribution of fidelities.
Since the distribution of fidelities of Haar random states is known to be $P_{\text{Haar}}(F) = (d-1)(1-F)^{d-2}$ for the Hilbert space of dimension $d$~\cite{Zyczkowski2005, Sim2019, Khatri2020}, the expected probability for each bin is given as 
\begin{align}
    \mathbb{P}_{\text{Haar}}(a \leq F < b) = \int_a^b (d-1)(1-F)^{d-2} d F = (1-a)^{d-1} - (1-b)^{d-1}.
\end{align}
The resulting experimentally evaluated expressivity is calculated as 
\begin{align}
    \mathrm{Expr} &:= D_{KL}\left( \hat{P}_V (F, \vec{\theta})\;\lVert\; P_{\text{Haar}}(F)\right)\\
    &= \sum_{[a_i, b_i)} \mathbb{P}_{\text{exp}}(a_i \leq F < b_i) \ln\left(\frac{\mathbb{P}_{\text{exp}}(a_i \leq F < b_i)}{\mathbb{P}_{\text{Haar}}(a_i \leq F < b_i)}\right),
\end{align}
where $\mathbb{P}_{\text{exp}}(a_i \leq F < b_i)$ is the experimentally evaluated probability for $F \in [a_i, b_i)$ and the sum is taken over all bins $[a_i, b_i)$. 
This measure of the expressivity is valued between $0$ for highly expressive circuits and $(d-1)\ln(N_\text{bins})$ for the lowest possible expressivity~\cite{Sim2019}.

\begin{figure}
    \centering
    \includegraphics[scale=1]{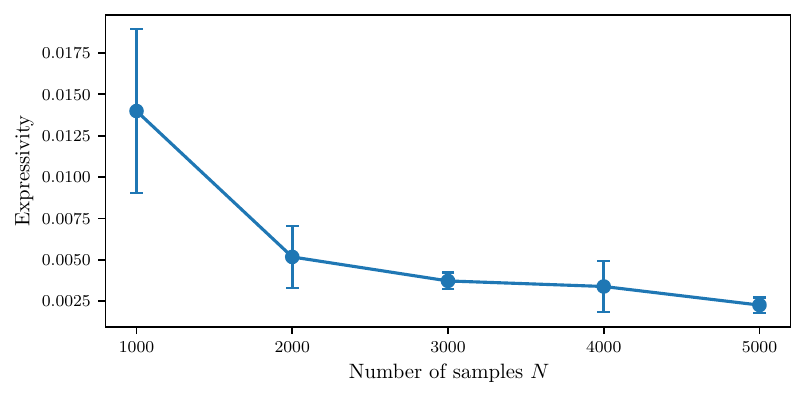}
    \caption{Evaluated expressivity for \circfour with $l=8$ layers depending on the number of fidelity-samples $N$.
    The markers indicate the average expressivity value across five evaluations, and the bars show the standard deviation.}
    \label{fig:expressivity_convergence}
\end{figure}
\begin{figure}
    \centering
    \includegraphics[scale=1]{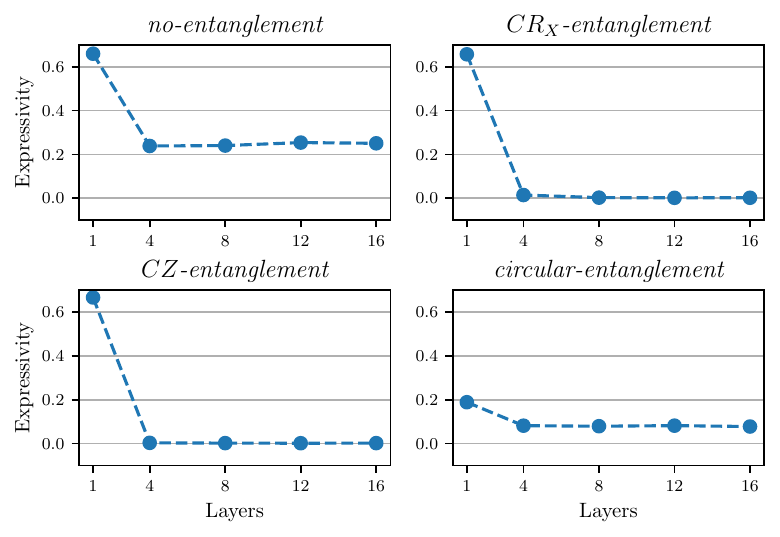}
    \caption{Expressivity of the circuits from \Cref{fig:circ_diagrams} used in \Cref{sec:exp}. Lower values correspond to a higher expressivity.}
    \label{fig:expressibilities}
\end{figure}
The obtained value for the expressivity is sensitive to the number of samples $N$~\cite{Sim2019}.
Therefore, as a preliminary test, we evaluated the convergence behavior of the expressivity calculation depending on $N$. 
An example of this evaluation is given in \Cref{fig:expressivity_convergence}, and we found that, similar to the case for four qubits in~\cite{Sim2019}, at $N=5000$ samples the evaluation generally converges.
Therefore, we use $5000$ samples each for the evaluation of the expressivity of the PQCs in our work.

\Cref{fig:expressibilities} shows the expressivity of the circuits that are used in our experiments. 
The PQCs \circfour and \circnine exhibit high expressivity at $l\geq 4$ with values close to zero, while \circone has lower expressivity and saturates at $\mathrm{Expr}\approx 0.25$.
Lastly, \circfifteen has higher expressivity at $l=1$ with $\mathrm{Expr}\approx 0.2$; however, increasing the number of layers does not considerably improve its expressivity as it saturates at $\mathrm{Expr}\approx 0.1$.


\begin{thebibliography}{63}
\providecommand{\natexlab}[1]{#1}
\providecommand{\url}[1]{\texttt{#1}}
\expandafter\ifx\csname urlstyle\endcsname\relax
  \providecommand{\doi}[1]{doi: #1}\else
  \providecommand{\doi}{doi: \begingroup \urlstyle{rm}\Url}\fi

\bibitem[Schuld and Petruccione(2018)]{Schuld2018}
Maria Schuld and Francesco Petruccione.
\newblock \emph{Supervised Learning with Quantum Computers}.
\newblock Quantum Science and Technology. Springer International Publishing,
  2018.
\newblock \doi{10.1007/978-3-319-96424-9}.

\bibitem[Biamonte et~al.(2017)Biamonte, Wittek, Pancotti, Rebentrost, Wiebe,
  and Lloyd]{Biamonte2017}
Jacob Biamonte, Peter Wittek, Nicola Pancotti, Patrick Rebentrost, Nathan
  Wiebe, and Seth Lloyd.
\newblock Quantum machine learning.
\newblock \emph{Nature}, 549\penalty0 (7671):\penalty0 195--202, 2017.
\newblock \doi{10.1038/nature23474}.

\bibitem[Huang et~al.(2021)Huang, Kueng, and Preskill]{Huang2021}
Hsin-Yuan Huang, Richard Kueng, and John Preskill.
\newblock Information-theoretic bounds on quantum advantage in machine
  learning.
\newblock \emph{Physical review letters}, 126 19:\penalty0 190505, 2021.
\newblock \doi{10.1103/PhysRevLett.126.190505}.

\bibitem[Rebentrost et~al.(2014)Rebentrost, Mohseni, and Lloyd]{Rebentrost2014}
Patrick Rebentrost, Masoud Mohseni, and Seth Lloyd.
\newblock Quantum support vector machine for big data classification.
\newblock \emph{Physical review letters}, 113\penalty0 (13):\penalty0 130503,
  2014.
\newblock \doi{10.1103/PhysRevLett.113.130503}.

\bibitem[Farhi and Neven(2018)]{Farhi2018}
Edward Farhi and Hartmut Neven.
\newblock Classification with quantum neural networks on near term processors.
\newblock \emph{arXiv:1802.06002}, 2018.
\newblock \doi{10.48550/arXiv.1802.06002}.

\bibitem[Farhi et~al.(2014)Farhi, Goldstone, and Gutmann]{Farhi2014}
Edward Farhi, Jeffrey Goldstone, and Sam Gutmann.
\newblock A quantum approximate optimization algorithm.
\newblock \emph{arXiv:1411.4028}, 2014.
\newblock \doi{10.48550/arXiv.1411.4028}.

\bibitem[Egger et~al.(2021)Egger, Mare{\v{c}}ek, and Woerner]{Egger_2021}
Daniel~J Egger, Jakub Mare{\v{c}}ek, and Stefan Woerner.
\newblock Warm-starting quantum optimization.
\newblock \emph{Quantum}, 5:\penalty0 479, 2021.
\newblock \doi{10.22331/q-2021-06-17-479}.

\bibitem[Poland et~al.(2020)Poland, Beer, and Osborne]{Poland2020}
Kyle Poland, Kerstin Beer, and Tobias~J. Osborne.
\newblock No free lunch for quantum machine learning.
\newblock \emph{arXiv:2003.14103}, 2020.
\newblock \doi{10.48550/arXiv.2003.14103}.

\bibitem[Sharma et~al.(2020)Sharma, Khatri, Cerezo, and Coles]{Sharma2020}
Kunal Sharma, Sumeet Khatri, M.~Cerezo, and Patrick~J. Coles.
\newblock Noise resilience of variational quantum compiling.
\newblock \emph{New Journal of Physics}, 22\penalty0 (4):\penalty0 043006,
  2020.
\newblock \doi{10.1088/1367-2630/ab784c}.

\bibitem[Khatri et~al.(2019)Khatri, LaRose, Poremba, Cincio, Sornborger, and
  Coles]{Khatri2019}
Sumeet Khatri, Ryan LaRose, Alexander Poremba, Lukasz Cincio, Andrew~T.
  Sornborger, and Patrick~J. Coles.
\newblock Quantum-assisted quantum compiling.
\newblock \emph{Quantum}, 3:\penalty0 140, 2019.
\newblock \doi{10.22331/q-2019-05-13-140}.

\bibitem[Benedetti et~al.(2019)Benedetti, Lloyd, Sack, and
  Fiorentini]{Benedetti2019a}
Marcello Benedetti, Erika Lloyd, Stefan Sack, and Mattia Fiorentini.
\newblock Parameterized quantum circuits as machine learning models.
\newblock \emph{Quantum Science and Technology}, 4\penalty0 (4):\penalty0
  043001, 2019.
\newblock \doi{10.1088/2058-9565/ab4eb5}.

\bibitem[Huang et~al.(2022)Huang, Broughton, Cotler, et~al.]{Huang2022}
Hsin-Yuan Huang, Michael Broughton, Jordan Cotler, et~al.
\newblock Quantum advantage in learning from experiments.
\newblock \emph{Science}, 376\penalty0 (6598):\penalty0 1182--1186, 2022.
\newblock \doi{10.1126/science.abn7293}.

\bibitem[Caro et~al.(2023)Caro, Huang, Ezzell, et~al.]{Caro2023}
Matthias~C. Caro, Hsin-Yuan Huang, Nicholas Ezzell, et~al.
\newblock Out-of-distribution generalization for learning quantum dynamics.
\newblock \emph{Nature Communications}, 14\penalty0 (1):\penalty0 3751, 2023.
\newblock \doi{10.1038/s41467-023-39381-w}.

\bibitem[Volkoff et~al.(2021)Volkoff, Holmes, and Sornborger]{Volkoff2021}
Tyler Volkoff, Zo\"e Holmes, and Andrew Sornborger.
\newblock Universal compiling and (no-)free-lunch theorems for
  continuous-variable quantum learning.
\newblock \emph{PRX Quantum}, 2:\penalty0 040327, 2021.
\newblock \doi{10.1103/PRXQuantum.2.040327}.

\bibitem[Sharma et~al.(2022)Sharma, Cerezo, Holmes, Cincio, Sornborger, and
  Coles]{Sharma2022}
Kunal Sharma, M.~Cerezo, Zoë Holmes, Lukasz Cincio, Andrew Sornborger, and
  Patrick~J. Coles.
\newblock Reformulation of the no-free-lunch theorem for entangled datasets.
\newblock \emph{Physical Review Letters}, 128\penalty0 (7), February 2022.
\newblock \doi{10.1103/physrevlett.128.070501}.

\bibitem[Wang et~al.(2024)Wang, Du, Tu, Luo, Yuan, and Tao]{Wang2023}
Xinbiao Wang, Yuxuan Du, Zhuozhuo Tu, Yong Luo, Xiao Yuan, and Dacheng Tao.
\newblock Transition role of entangled data in quantum machine learning.
\newblock \emph{Nature Communications}, 15\penalty0 (1):\penalty0 3716, 2024.
\newblock \doi{10.1038/s41467-024-47983-1}.

\bibitem[Mandl et~al.(2023{\natexlab{a}})Mandl, Barzen, Leymann, and
  Vietz]{Mandl2023reducing}
Alexander Mandl, Johanna Barzen, Frank Leymann, and Daniel Vietz.
\newblock On reducing the amount of samples required for training of {QNNs}:
  Constraints on the linear structure of the training data.
\newblock \emph{arXiv:2309.13711}, 2023{\natexlab{a}}.
\newblock \doi{10.48550/arXiv.2309.13711}.

\bibitem[Holmes et~al.(2022)Holmes, Sharma, Cerezo, and Coles]{Holmes2022}
Zo{\"e} Holmes, Kunal Sharma, M.~Cerezo, and Patrick~J. Coles.
\newblock Connecting ansatz expressibility to gradient magnitudes and barren
  plateaus.
\newblock \emph{PRX Quantum}, 3\penalty0 (1):\penalty0 010313, 2022.
\newblock \doi{10.1103/PRXQuantum.3.010313}.

\bibitem[Larocca et~al.(2025)Larocca, Thanasilp, Wang, et~al.]{Larocca2025}
Martín Larocca, Supanut Thanasilp, Samson Wang, et~al.
\newblock Barren plateaus in variational quantum computing.
\newblock \emph{Nature Reviews Physics}, 2025.
\newblock \doi{10.1038/s42254-025-00813-9}.

\bibitem[Cerezo et~al.(2021{\natexlab{a}})Cerezo, Sone, Volkoff, Cincio, and
  Coles]{Cerezo2021a}
M.~Cerezo, Akira Sone, Tyler Volkoff, Lukasz Cincio, and Patrick~J. Coles.
\newblock Cost function dependent barren plateaus in shallow parametrized
  quantum circuits.
\newblock \emph{Nature Communications}, 12\penalty0 (1):\penalty0 1791,
  2021{\natexlab{a}}.
\newblock \doi{10.1038/s41467-021-21728-w}.

\bibitem[Wang et~al.(2021)Wang, Fontana, Cerezo, et~al.]{Wang2021}
Samson Wang, Enrico Fontana, M.~Cerezo, et~al.
\newblock Noise-induced barren plateaus in variational quantum algorithms.
\newblock \emph{Nature Communications}, 12\penalty0 (1):\penalty0 6961, 2021.
\newblock \doi{10.1038/s41467-021-27045-6}.

\bibitem[Leone et~al.(2024)Leone, Oliviero, Cincio, and Cerezo]{Leone2024}
Lorenzo Leone, Salvatore~FE Oliviero, Lukasz Cincio, and M.~Cerezo.
\newblock On the practical usefulness of the hardware efficient ansatz.
\newblock \emph{Quantum}, 8:\penalty0 1395, 2024.
\newblock \doi{10.22331/q-2024-07-03-1395}.

\bibitem[Thanasilp et~al.(2023)Thanasilp, Wang, Nghiem, Coles, and
  Cerezo]{Thanasilp2023}
Supanut Thanasilp, Samson Wang, Nhat~Anh Nghiem, Patrick Coles, and M.~Cerezo.
\newblock Subtleties in the trainability of quantum machine learning models.
\newblock \emph{Quantum Machine Intelligence}, 5\penalty0 (1), 2023.
\newblock \doi{10.1007/s42484-023-00103-6}.

\bibitem[Stiliadou et~al.(2024)Stiliadou, Barzen, Leymann, Mandl, and
  Weder]{Stiliadou2024}
Lavinia Stiliadou, Johanna Barzen, Frank Leymann, Alexander Mandl, and Benjamin
  Weder.
\newblock Exploring the cost landscape of variational quantum algorithms.
\newblock In \emph{Symposium and Summer School on Service-Oriented Computing},
  pages 128--142. Springer, 2024.
\newblock \doi{10.1007/978-3-031-72578-4_7}.

\bibitem[Arrasmith et~al.(2022)Arrasmith, Holmes, Cerezo, and
  Coles]{Arrasmith2022}
Andrew Arrasmith, Zoë Holmes, M.~Cerezo, and Patrick~J. Coles.
\newblock Equivalence of quantum barren plateaus to cost concentration and
  narrow gorges.
\newblock \emph{Quantum Science and Technology}, 7\penalty0 (4):\penalty0
  045015, 2022.
\newblock \doi{10.1088/2058-9565/ac7d06}.

\bibitem[Mhiri et~al.(2025)Mhiri, Puig, Lerch, Rudolph, Chotibut, Thanasilp,
  and Holmes]{Mhiri2025}
Hela Mhiri, Ricard Puig, Sacha Lerch, Manuel~S Rudolph, Thiparat Chotibut,
  Supanut Thanasilp, and Zo{\"e} Holmes.
\newblock A unifying account of warm start guarantees for patches of quantum
  landscapes.
\newblock \emph{arXiv:2502.07889}, 2025.
\newblock \doi{10.48550/arXiv.2502.07889}.

\bibitem[Mastropietro et~al.(2023)Mastropietro, Korpas, Kungurtsev, and
  Marecek]{Mastropietro2023}
Daniel Mastropietro, Georgios Korpas, Vyacheslav Kungurtsev, and Jakub Marecek.
\newblock Parallel variational quantum algorithms with gradient-informed
  restart to speed up optimisation in the presence of barren plateaus.
\newblock \emph{arXiv:2311.18090}, 2023.
\newblock \doi{10.48550/arXiv.2311.18090}.

\bibitem[Nielsen and Chuang(2010)]{nielsen2002quantum}
Michael~A. Nielsen and Isaac~L. Chuang.
\newblock \emph{{Quantum Computation and Quantum Information}}.
\newblock Cambridge University Press, 2010.
\newblock \doi{10.1017/CBO9780511976667}.

\bibitem[Bennett et~al.(1996)Bennett, Bernstein, Popescu, and
  Schumacher]{Bennett1996}
Charles~H Bennett, Herbert~J Bernstein, Sandu Popescu, and Benjamin Schumacher.
\newblock Concentrating partial entanglement by local operations.
\newblock \emph{Physical Review A}, 53\penalty0 (4):\penalty0 2046, 1996.
\newblock \doi{10.1103/PhysRevA.53.2046}.

\bibitem[Bengtsson and {\.Z}yczkowski(2017)]{Bengtsson2017}
Ingemar Bengtsson and Karol {\.Z}yczkowski.
\newblock \emph{Geometry of quantum states: an introduction to quantum
  entanglement}.
\newblock Cambridge University Press, 2017.
\newblock \doi{10.1017/CBO9780511535048}.

\bibitem[Mele(2024)]{Mele2023}
Antonio~Anna Mele.
\newblock Introduction to haar measure tools in quantum information: A
  beginner's tutorial.
\newblock \emph{Quantum}, 8:\penalty0 1340, 2024.
\newblock \doi{10.22331/q-2024-05-08-1340}.

\bibitem[Pucha{\l}a and Miszczak(2017)]{Puchala_2017}
Zbigniew Pucha{\l}a and Jaros{\l}aw~A. Miszczak.
\newblock Symbolic integration with respect to the {Haar} measure on the
  unitary groups.
\newblock \emph{Bulletin of the Polish Academy of Sciences Technical Sciences},
  65\penalty0 (1):\penalty0 21--27, 2017.
\newblock \doi{10.1515/bpasts-2017-0003}.

\bibitem[Wu and Yu(2020)]{Wu2020}
Shao-xiong Wu and Chang-shui Yu.
\newblock Quantum speed limit based on the bound of {Bures} angle.
\newblock \emph{Scientific reports}, 10\penalty0 (1):\penalty0 5500, 2020.
\newblock \doi{10.1038/s41598-020-62409-w}.

\bibitem[Hall(2013)]{Hall2013}
Brian~C Hall.
\newblock Lie groups, {Lie} algebras, and representations.
\newblock In \emph{Quantum Theory for Mathematicians}, pages 333--366.
  Springer, 2013.
\newblock \doi{10.1007/978-3-319-13467-3}.

\bibitem[Haah et~al.(2023)Haah, Kothari, O’Donnell, and Tang]{Haah2023}
Jeongwan Haah, Robin Kothari, Ryan O’Donnell, and Ewin Tang.
\newblock Query-optimal estimation of unitary channels in diamond distance.
\newblock In \emph{2023 IEEE 64th Annual Symposium on Foundations of Computer
  Science (FOCS)}, pages 363--390. IEEE, 2023.
\newblock \doi{10.1109/FOCS57990.2023.00028}.

\bibitem[Sim et~al.(2019)Sim, Johnson, and Aspuru-Guzik]{Sim2019}
Sukin Sim, Peter~D. Johnson, and Al{\'{a}}n Aspuru-Guzik.
\newblock Expressibility and entangling capability of parameterized quantum
  circuits for hybrid quantum-classical algorithms.
\newblock \emph{Advanced Quantum Technologies}, 2\penalty0 (12):\penalty0
  1900070, 2019.
\newblock \doi{10.1002/qute.201900070}.

\bibitem[Cerezo et~al.(2021{\natexlab{b}})]{Cerezo2021}
M.~Cerezo et~al.
\newblock Variational quantum algorithms.
\newblock \emph{Nature Reviews Physics}, 3\penalty0 (9):\penalty0 625--644,
  2021{\natexlab{b}}.
\newblock \doi{10.1038/s42254-021-00348-9}.

\bibitem[Barzen and Leymann(2025)]{Barzen2025}
Johanna Barzen and Frank Leymann.
\newblock On the differential topology of expressivity of parameterized quantum
  circuits.
\newblock \emph{AppliedMath}, 5\penalty0 (3):\penalty0 121, 2025.
\newblock \doi{10.3390/appliedmath5030121}.

\bibitem[Guo et~al.(2024)Guo, Muta, and Zhao]{Guo2024}
Xiaoyu Guo, Takahiro Muta, and Jianjun Zhao.
\newblock Quantum circuit ansatz: patterns of abstraction and reuse of quantum
  algorithm design.
\newblock In \emph{2024 IEEE International Conference on Quantum Software
  (QSW)}, pages 69--80. IEEE, 2024.
\newblock \doi{10.1109/QSW62656.2024.00021}.

\bibitem[Wang et~al.(2025)Wang, Wang, Li, et~al.]{Wang2024}
Shengbin Wang, Peng Wang, Guihui Li, et~al.
\newblock Variational quantum eigensolver with linear depth problem-inspired
  ansatz for solving portfolio optimization in finance.
\newblock \emph{Science China Information Sciences}, 68\penalty0 (8):\penalty0
  1--11, 2025.
\newblock \doi{10.1007/s11432-024-4185-1}.

\bibitem[Kandala et~al.(2017)Kandala, Mezzacapo, Temme, et~al.]{Kandala2017}
Abhinav Kandala, Antonio Mezzacapo, Kristan Temme, et~al.
\newblock Hardware-efficient variational quantum eigensolver for small
  molecules and quantum magnets.
\newblock \emph{Nature}, 549\penalty0 (7671):\penalty0 242--246, 2017.
\newblock \doi{10.1038/nature23879}.

\bibitem[Funcke et~al.(2021)Funcke, Hartung, Jansen, K{\"u}hn, and
  Stornati]{Funcke2021}
Lena Funcke, Tobias Hartung, Karl Jansen, Stefan K{\"u}hn, and Paolo Stornati.
\newblock Dimensional expressivity analysis of parametric quantum circuits.
\newblock \emph{Quantum}, 5:\penalty0 422, 2021.

\bibitem[Friedrich and Maziero(2023)]{Friedrich2023}
Lucas Friedrich and Jonas Maziero.
\newblock Quantum neural network cost function concentration dependency on the
  parametrization expressivity.
\newblock \emph{Scientific reports}, 13\penalty0 (1):\penalty0 9978, 2023.
\newblock \doi{10.1038/s41598-023-37003-5}.

\bibitem[Wu et~al.(2023)Wu, Jian, Du, Chen, and Zhou]{Wu2023}
Keke Wu, Xiangru Jian, Rui Du, Jingrun Chen, and Xiang Zhou.
\newblock Roughness index for loss landscapes of neural network models of
  partial differential equations.
\newblock In \emph{2023 IEEE International Conference on Big Data (BigData)},
  pages 966--975. IEEE, 2023.
\newblock \doi{10.1109/BigData59044.2023.10386909}.

\bibitem[Pesah et~al.(2021)Pesah, Cerezo, Wang, Volkoff, Sornborger, and
  Coles]{Pesah2021}
Arthur Pesah, M.~Cerezo, Samson Wang, Tyler Volkoff, Andrew~T. Sornborger, and
  Patrick~J. Coles.
\newblock Absence of barren plateaus in quantum convolutional neural networks.
\newblock \emph{Physical Review X}, 11\penalty0 (4), 2021.
\newblock \doi{10.1103/physrevx.11.041011}.

\bibitem[McClean et~al.(2018)McClean, Boixo, Smelyanskiy, Babbush, and
  Neven]{McClean2018}
Jarrod~R McClean, Sergio Boixo, Vadim~N Smelyanskiy, Ryan Babbush, and Hartmut
  Neven.
\newblock Barren plateaus in quantum neural network training landscapes.
\newblock \emph{Nature Communications}, 9\penalty0 (1):\penalty0 4812, 2018.
\newblock \doi{10.1038/s41467-018-07090-4}.

\bibitem[Grant et~al.(2019)Grant, Wossnig, Ostaszewski, and
  Benedetti]{Grant2019}
Edward Grant, Leonard Wossnig, Mateusz Ostaszewski, and Marcello Benedetti.
\newblock An initialization strategy for addressing barren plateaus in
  parametrized quantum circuits.
\newblock \emph{{Quantum}}, 3:\penalty0 214, 2019.
\newblock \doi{10.22331/q-2019-12-09-214}.

\bibitem[Crognaletti et~al.(2024)Crognaletti, Grossi, and
  Bassi]{Crognaletti2024}
Giulio Crognaletti, Michele Grossi, and Angelo Bassi.
\newblock Estimates of loss function concentration in noisy parametrized
  quantum circuits.
\newblock \emph{arXiv preprint arXiv:2410.01893}, 2024.
\newblock \doi{10.48550/arXiv.2410.01893}.

\bibitem[Mandl et~al.(2023{\natexlab{b}})Mandl, Barzen, Bechtold, Leymann, and
  Stiliadou]{Mandl2025datarepo}
Alexander Mandl, Johanna Barzen, Marvin Bechtold, Frank Leymann, and Lavinia
  Stiliadou.
\newblock Data repository for {'Loss Behavior in Supervised Learning With
  Entangled States'}.
\newblock 2023{\natexlab{b}}.
\newblock \doi{10.18419/DARUS-5174}.

\bibitem[Paszke et~al.(2019)Paszke, Gross, Massa, et~al.]{Paszke2019}
Adam Paszke, Sam Gross, Francisco Massa, et~al.
\newblock Pytorch: An imperative style, high-performance deep learning library.
\newblock In \emph{Advances in Neural Information Processing Systems},
  volume~32. Curran Associates, Inc., 2019.

\bibitem[Kraft(1988)]{Kraft1988}
Dieter Kraft.
\newblock A software package for sequential quadratic programming.
\newblock \emph{Forschungsbericht- Deutsche Forschungs- und Versuchsanstalt fur
  Luft- und Raumfahrt}, 1988.

\bibitem[Virtanen et~al.(2020)Virtanen, Gommers, Oliphant,
  et~al.]{Virtanen2020_scipy}
Pauli Virtanen, Ralf Gommers, Travis~E Oliphant, et~al.
\newblock {{SciPy} 1.0: Fundamental Algorithms for Scientific Computing in
  Python}.
\newblock \emph{Nature Methods}, 17:\penalty0 261--272, 2020.
\newblock \doi{10.1038/s41592-019-0686-2}.

\bibitem[Barenco et~al.(1995)Barenco, Bennett, Cleve, et~al.]{Barenco1995}
Adriano Barenco, Charles~H Bennett, Richard Cleve, et~al.
\newblock Elementary gates for quantum computation.
\newblock \emph{Physical review A}, 52\penalty0 (5):\penalty0 3457, 1995.
\newblock \doi{10.1103/PhysRevA.52.3457}.

\bibitem[Truger et~al.(2024)Truger, Barzen, Bechtold, et~al.]{Truger2023}
Felix Truger, Johanna Barzen, Marvin Bechtold, et~al.
\newblock Warm-starting and quantum computing: A systematic mapping study.
\newblock \emph{ACM Comput. Surv.}, 56\penalty0 (9), 2024.
\newblock \doi{10.1145/3652510}.

\bibitem[Cerezo et~al.(2023)Cerezo, Larocca, Garc{\'\i}a-Mart{\'\i}n,
  et~al.]{Cerezo2023}
M.~Cerezo, Martin Larocca, Diego Garc{\'\i}a-Mart{\'\i}n, et~al.
\newblock Does provable absence of barren plateaus imply classical
  simulability?{O}r, why we need to rethink variational quantum computing.
\newblock \emph{arXiv preprint arXiv:2312.09121}, 2023.
\newblock \doi{10.48550/arXiv.2312.09121}.

\bibitem[Tromborg and Waldenstr{\o}m(1978)]{Tromborg1978}
B~Tromborg and S~Waldenstr{\o}m.
\newblock Bounds on the diagonal elements of a unitary matrix.
\newblock \emph{Linear Algebra and its Applications}, 20\penalty0 (3):\penalty0
  189--195, 1978.
\newblock \doi{10.1016/0024-3795(78)90017-4}.

\bibitem[Liesen and Mehrmann(2015)]{Liesen2015}
Jörg Liesen and Volker Mehrmann.
\newblock \emph{Linear Algebra}.
\newblock Springer, 2015.
\newblock \doi{10.1007/978-3-319-24346-7}.

\bibitem[Gelfand and Saul(2001)]{Gelfand2001}
I.~M. Gelfand and Mark Saul.
\newblock \emph{Trigonometry}.
\newblock Birkhäuser Boston, 2001.
\newblock ISBN 9781461201496.
\newblock \doi{10.1007/978-1-4612-0149-6}.

\bibitem[Bagul and Panchal(2022)]{Bagul2022}
Yogesh~J Bagul and Satish~K Panchal.
\newblock Certain inequalities of kober and lazarevic type.
\newblock \emph{The Journal of the Indian Mathematical Society}, 89\penalty0
  (1-2):\penalty0 01--07, 2022.
\newblock \doi{10.18311/jims/2022/20737}.

\bibitem[Qi(2006)]{Qi2006}
Feng Qi.
\newblock Jordan's inequality: refinements, generalizations, applications and
  related problems.
\newblock \emph{Research report collection}, 9\penalty0 (3), 2006.

\bibitem[Nakaji and Yamamoto(2021)]{Nakaji2021}
Kouhei Nakaji and Naoki Yamamoto.
\newblock Expressibility of the alternating layered ansatz for quantum
  computation.
\newblock \emph{Quantum}, 5:\penalty0 434, 2021.
\newblock \doi{10.22331/q-2021-04-19-434}.

\bibitem[\ifmmode~\dot{Z}\else \.{Z}\fi{}yczkowski and
  Sommers(2005)]{Zyczkowski2005}
Karol \ifmmode~\dot{Z}\else \.{Z}\fi{}yczkowski and Hans-J\"urgen Sommers.
\newblock Average fidelity between random quantum states.
\newblock \emph{Phys. Rev. A}, 71:\penalty0 032313, Mar 2005.
\newblock \doi{10.1103/PhysRevA.71.032313}.

\bibitem[Khatri(2020)]{Khatri2020}
Sumeet Khatri.
\newblock Random pure states.
\newblock Notes available online \url{https://sumeetkhatri.com/notes/}, 2020.
\newblock Accessed 2025-09-09.

\end{thebibliography}
\end{document}